\documentclass[11pt]{amsart}
\usepackage{amsthm}

\usepackage[margin=1in]{geometry}

\usepackage{graphicx}
\usepackage{amsfonts,amsmath,amssymb}

\usepackage{multirow}
\usepackage[noend]{algorithmic}
\usepackage{xspace}
\usepackage{footnote}
\usepackage{float}
\usepackage[nameinlink,capitalize]{cleveref}
\newtheorem{theorem}{Theorem}[section]
\newtheorem{proposition}[theorem]{Proposition}
\newtheorem{lemma}[theorem]{Lemma}

\newtheorem{definition}[theorem]{Definition}

\makesavenoteenv{tabular}
\makesavenoteenv{table}

\newcommand{\complexityclass}[2][]{\ensuremath{\mathsf{#2}\ifthenelse{\isempty{#1}}{}{(#1)}}}

\newcommand{\bE}{\ensuremath{\mathbf{E}}}

\DeclareMathOperator{\polyloglog}{polyloglog}
\DeclareMathOperator{\polylog}{polylog}
\DeclareMathOperator{\poly}{poly}

\newtheorem*{rep@theorem}{\rep@title}
\newcommand{\newreptheorem}[2]{%
\newenvironment{rep#1}[1]{%
 \def\rep@title{#2 \ref{##1}}%
 \begin{rep@theorem}}%
 {\end{rep@theorem}}}
\makeatother

\newreptheorem{theorem}{Theorem}

\begin{document}

\begin{abstract}
We describe approximation algorithms in Linial's classic LOCAL model of distributed computing to find maximum-weight matchings in a hypergraph of rank $r$. Our main result is a deterministic algorithm to generate a matching which is an $O(r)$-approximation to the maximum weight matching, running in $\tilde O(r \log \Delta + \log^2 \Delta + \log^* n)$ rounds. (Here, the $\tilde O()$ notations hides $\polyloglog \Delta$ and $\polylog r$ factors). This is based on a number of new derandomization techniques extending methods of Ghaffari, Harris \& Kuhn (2017).

As a main application, we obtain nearly-optimal algorithms for the long-studied problem of maximum-weight graph matching. Specifically, we get $(1+\epsilon)$ approximation algorithm in $\tilde O(\frac{\log \Delta}{\epsilon^3}+ \polylog(1/\epsilon, \log \log n))$ randomized time and $\tilde O( \frac{\log^2 \Delta}{ \epsilon^4} + \frac{\log^*n}{\epsilon})$ deterministic time. 

  The second application is a faster algorithm for hypergraph maximal matching, a versatile subroutine introduced in Ghaffari et al. (2017) for a variety of local graph algorithms. This  gives an algorithm for $(2 \Delta - 1)$-edge-list coloring in $\tilde O(\log^2 \Delta \log n)$ rounds deterministically or $\tilde O( (\log \log n)^3 )$ rounds randomly. Another consequence (with additional optimizations) are algorithms which generates an edge-orientation with out-degree at most $\lceil (1+\epsilon) \lambda \rceil$ for a graph of arboricity $\lambda$; for fixed $\epsilon$ this takes $\tilde O(\log^6 n)$ deterministic time and or $\tilde O(\log^3 n )$ randomized time.
\end{abstract}

\title[Approx algorithms for max matchings in graphs and hypergraphs]{Distributed local approximation algorithms for maximum matching in graphs and hypergraphs}

\author[David G. Harris]{
{\sc David G.~Harris}$^{1}$
}

\setcounter{footnote}{0}
\addtocounter{footnote}{1}
\footnotetext{Department of Computer Science, University of Maryland, 
College Park, MD 20742. 
Email: \texttt{davidgharris29@gmail.com}.}

\maketitle

This is an extended version of a paper appearing in the 60th annual IEEE Symposium on Foundations of Computer Science (FOCS), 2019.

\section{Introduction}
Consider a hypergraph $H = (V,E)$ with rank (maximum edge size) $r$ and maximum degree $\Delta$. A matching of $H$ is a set of pairwise-disjoint edges; equivalently, it is an independent set of the line graph of $H$.  We develop distributed hypergraph matching algorithms in Linial's classic LOCAL model of computation \cite{lin92}. In this model, time proceeds in synchronous rounds, in which each vertex in the hypergraph can communicate with any other vertex sharing a hyperedge. Computation and message size are unbounded.

There are two main reasons for studying hypergraph matching in this context. First, many symmetry-breaking and locality issues for graphs remain relevant to hypergraphs. Since graph maximal matching is one of the ``big four'' symmetry-breaking problems (which also includes maximal independent set (MIS), vertex coloring, and edge coloring) \cite{pan-rizzi}, it is a natural extension to generalize it to the richer setting of hypergraphs.

The second and more important reason is that distributed hypergraph matching can be used as a clean subroutine for a number of graph algorithms \cite{fgk}. These include maximum-weight matching,  edge-coloring and Nash-Williams decomposition \cite{prev,GKMU18}. Such algorithms need to find many disjoint ``augmenting paths'' (in various flavors) in the graph. These paths can be represented by an auxiliary hypergraph, and a disjoint collection of such paths corresponds to a hypergraph matching.

The problem of \emph{Hypergraph Maximum Weight Matching (HMWM)} is to find a matching $M$ whose weight $a(M) = \sum_{e \in M} a(e)$ is maximum for a given edge-weighting function $a: E \rightarrow [0,\infty)$. This is often intractable to solve exactly, so we define a \emph{$\rho$-approximation for HMWM} to be a matching $M$ whose weight $a(M)$ is at least $1/\rho$ times the maximum matching weight. The main contribution of this paper is a deterministic LOCAL algorithm to approximate HMWM. 
\begin{theorem}[Simplified]
\label{main-thm1}
There is a deterministic $\tilde O( r \log \Delta + \log^2 \Delta + \log^* n)$-round  algorithm for $O(r)$-approximation to HMWM.\footnote{  Throughout, we define  $\tilde O(x)$ to be $x \polylog(x)$.}
\end{theorem}

Randomization reduces the run-time still further and yields a truly local algorithm:
\begin{theorem}
\label{main-thm2}
There is a randomized $\tilde O( \log \Delta + r \log \log \tfrac{1}{\delta} + ( \log \log \tfrac{1}{\delta})^2 )$-round  algorithm for $O(r)$-approximation to HMWM with success probability at least $1 - \delta$.
\end{theorem}

By far the most important application of HMWM is to \emph{Graph Maximum Weight Matching (GMWM)}. This is one of the most well-studied problems in algorithmic graph theory, with many variants such as specialized graph classes or computational models. We cannot summarize the full literature here, but let us provide a brief summary of the role played by our HMWM algorithm.

\subsection{Graph matching} Overall, there are three main paradigms for GMWM approximation algorithms. The first is based on \emph{maximal matching}, which is a $2$-approximation to (unweighted) graph maximum matching. With some extensions, this can also be used for approximation algorithms in weighted graphs \cite{lotker}. The best algorithms for maximal matching are a deterministic algorithm in $O( \log^2 \Delta \log n)$ rounds \cite{fischer}, and a randomized algorithm in $O( \log \Delta +  (\log \log n)^3)$ rounds (a combination of a randomized algorithm of \cite{beps} with the deterministic algorithm of \cite{fischer}).
 
 The second paradigm is based on rounding a fractional matching. For general graphs, there can be a gap of $3/2$ between the maximum weight matching and the fractional matching LP (to be distinguished from the matching polytope). Hence these algorithms typically achieve an approximation ratio of $3/2$ at best, sometimes with additional loss in the rounding. The most recent example is the deterministic algorithm \cite{ahmadi2} for $(3/2 + \epsilon)$-approximate maximum matching in $O(\log W+ \log^2 \Delta + \log^*n)$ rounds for fixed $\epsilon$, where $W$ is the ratio between maximum and minimum edge weight. Note that, unlike maximal matching, the run-time has a negligible dependence on $n$.  (This algorithm has better approximation ratios on restricted graph classes; for example, a $(1+\epsilon)$ approximation for bipartite graphs.)
 
 To get an approximation factor arbitrarily close to one in general graphs, it appears necessary to use the third type of approximation algorithm based on \emph{path augmentation}. These algorithms build up the matching by iteratively finding and applying short augmenting paths. The task of finding these paths can be formulated as a hypergraph matching problem.

 Our focus here will be on the $(1+\epsilon)$-approximate GMWM problem in the LOCAL model. In addition to the run-time, there are some other important properties to keep in mind. The first is how the algorithm depends on the dynamic range $W$ of the edge weights. Some algorithms only work in the case $W = 1$, i.e. maximum cardinality matching. We refer to these as \emph{unweighted} algorithms. Other algorithms may have a run-time scaling logarithmically in $W$. 

The second property is the message size. Our focus is on the LOCAL graph model, in which message size is unbounded. A more restrictive model CONGEST is often used, in which message sizes are limited to $O(\log n)$ bits per edge per round. 

A third property is the role of randomness. We have the usual dichotomy in local algorithms between randomized and deterministic. It is traditional in randomized algorithms to seek success probability $1 - 1/\poly(n)$, known as \emph{with high probability} (w.h.p.). Some GMWM algorithms give a weaker guarantee that the algorithm returns a matching $M$ whose \emph{expected weight} is within a $(1+\epsilon)$ factor of the maximum weight. Note that when $W$ is large, we cannot expect any meaningful concentration for the matching weight, since a single edge may contribute disproportionately. We refer to this as a \emph{first-moment} probabilistic guarantee.

Table~\ref{fi1} summarizes a number of $(1+\epsilon)$-approximation GMWM algorithms, listed roughly in order of publication. For readability, we have simplified the run-time bounds, omitting some asymptotic notations as well as dropping some second-order terms. We note that more recent MIS algorithms \cite{ghaffari-mis, ghaffari-roz} could potentially be used here; since such results are not trivial and are not claimed in the literature, we do not include them in the table.
\begin{table*}[h]
\centering
\begin{tabular}{|l|l|l|l|l|l|}
\hline
Ref & Randomization & Message size & run-time & Weighted? \\
\hline
\hline
\cite{czy3} & Det & LOCAL & $(\log n)^{O_{\epsilon}(1)}$ & No \\
\hline
\cite{lotker2} & W.h.p. & LOCAL & $\epsilon^{-3} \log n$ & No \\
\hline
\cite{lotker2} & W.h.p. & CONGEST & $2^{1/\epsilon} \log n$ & No \\
\hline
\cite{nieberg} & W.h.p. & LOCAL & $\epsilon^{-3} \log n$ & Yes \\
\hline
\cite{even}& Det & LOCAL & $\Delta^{1/\epsilon} + \epsilon^{-2} \log^* n$  & No \\
\hline
\cite{even} & Det & LOCAL & $( \log W n)^{1/\epsilon} ( \Delta^{1/\epsilon} + \log^* n)$ & Yes \\
  \hline
  \cite{bcgs} & First-moment & CONGEST & $2^{1/\epsilon} \frac{\log \Delta}{\log \log \Delta}$ & No \\
  \hline
  \cite{fgk} & Det & CONGEST & $\Delta^{1/\epsilon} + \poly(1/\epsilon) \log^* n $ & No \\
  \hline
  \cite{prev} & Det & LOCAL & $\epsilon^{-9} \log^5 \Delta \log^2 n$ & No \\
  \hline
  \cite{GKMU18} + \cite{prev} & Det & LOCAL & $\epsilon^{-7} \log^4 \Delta \log^3 n$  (for $\Delta \ll n$) & Yes \\
  \hline
  \hline
  This paper & Det & LOCAL & $\epsilon^{-4} \log^2 \Delta + \epsilon^{-1} \log^* n$ & Yes \\
  \hline
  This paper & W.h.p. & LOCAL & $\epsilon^{-3} \log \Delta + \epsilon^{-3} \log \log n + \epsilon^{-2} (\log \log n)^2$  & Yes \\
  \hline
  This paper & First-moment & LOCAL & $\epsilon^{-3} \log \Delta$ & Yes \\
  \hline
\end{tabular}
\vspace{1mm}
\caption{Comparison of graph $(1+\epsilon)$-approximate maximum matching algorithms.}
\label{fi1}
\vspace{-1mm}
\end{table*}

Our HMWM algorithm yields a $(1+\epsilon)$-approximation algorithm for graph matching:
\begin{theorem}
\label{mt1}
For $\epsilon > 0$, there is a $\tilde O(\epsilon^{-4} \log^2 \Delta + \epsilon^{-1} \log^* n)$-round deterministic algorithm for  $(1+\epsilon)$-approximate GMWM.
  
For $\delta > 0$ there is $\tilde O(\epsilon^{-3} \log \Delta + \epsilon^{-3} \log \log \tfrac{1}{\delta} + \epsilon^{-2} (\log \log \tfrac{1}{\delta})^2)$-round randomized algorithm for $(1+\epsilon)$-approximate GMWM with success probability at least $1 - \delta$.
\end{theorem}

These are the first $(1+\epsilon)$-approximation algorithms (either randomized or deterministic) that simultaneously have three desirable properties: (1) run-time essentially independent of $n$; (2) a polynomial dependence on $1/\epsilon$, and (3) allowing weighted graphs, without run-time dependence on the parameter $W$.   Note that the deterministic algorithm  matches the fastest prior \emph{constant-factor} approximation algorithms  \cite{ahmadi2, fischer}, and both the randomized and deterministic algorithms nearly match a number of known lower bounds.

\subsection{Maximal matching and other applications}
The HMWM algorithm can also be used for the closely related problem of hypergraph \emph{maximal} matching (HMM):
\begin{theorem}[Simplified]
\label{mt2}
HMM can be solved in $\tilde O( (\log n) (r^2 \log \Delta + r \log^2 \Delta) )$ rounds deterministically  or  $\tilde O(r \log^2 \Delta + r^2 (\log \log n)^2 + r (\log \log n)^3)$ rounds w.h.p.
\end{theorem}
By contrast, the deterministic algorithm of \cite{prev} uses $O(r^2 \log(n \Delta) \log n \log^4 \Delta)$ rounds. There is a lower bound for graph maximal matching algorithms which have sublinear dependence on $\Delta$ of $\Omega(\frac{\log n}{\log \log n})$ deterministic rounds or $\Omega(\frac{\log \log n}{\log \log \log n})$ randomized rounds \cite{bbhors}. Thus, our algorithms have nearly optimal dependence on $n$, up to $\polyloglog n$ factors.

Using Theorem~\ref{mt2} as a subroutine, we immediately get improved distributed algorithms for a number of graph problems. Three simple consequences are for edge coloring:
\begin{theorem}
  \label{mt4} 
Let $G$ be a graph with maximum degree $\Delta$.
\begin{enumerate}
\item There is a $\tilde O(\log n \log^2 \Delta)$-round deterministic algorithm for $(2\Delta-1)$-list-edge-coloring.
\item There is a $\tilde O( (\log  \log n)^3 )$-round randomized algorithm  for $(2\Delta-1)$-list-edge-coloring
\item There is a $\tilde O(\Delta^4 \log^6 n)$-round deterministic algorithm for $\tfrac{3}{2} \Delta$-edge-coloring. 
\end{enumerate}
\end{theorem}

One more involved application of HMM (along with some additional optimizations) is the following approximate Nash-Williams graph decomposition:
\begin{theorem}
\label{mt3}
For a multi-graph $G$ of arboricity $\lambda$  there is a $\tilde O(\log^6 n / \epsilon^4)$-round deterministic algorithm and a $\tilde O(\log^3 n/\epsilon^3)$-round randomized algorithm to find an edge-orientation of $G$ with maximum out-degree $\lceil (1+\epsilon) \lambda \rceil$.
\end{theorem}

\subsection{Overview and outline}
Our algorithm is based on a simple randomized rounding procedure we refer to as \emph{direct rounding}. Suppose we are given a maximum-weight \emph{fractional} matching $h$ for some edge-weighting function $a$ (this can be found using an LP solving procedure of \cite{kmw}).  Consider the process wherein each edge is selected independently with probability $h(e) \log r$; for any vertex with more than $\Delta' = \Omega( \log r)$ selected edges,  we discard all such edges.

One can easily check that $L$ has expected weight of $\Omega( \sum_{e \in E} a(e) h(e) \log r )$. Since $L$ has small degree, it is inexpensive to convert it into a matching of weight $\Omega( \frac{a(L)}{\Delta' r} ) = \Omega( \frac{\sum_{e \in E} a(e) h(e)}{r})$, which is an $O(r)$-approximation to HMWM.

The crux of our algorithm, and the most important technical contribution of this paper, is derandomization of direct rounding. We use three main derandomization techniques. These all build on each other, and may be of interest in other settings.

In Section~\ref{rnd-sec}, we describe the first general technique for derandomizing  LOCAL graph algorithms. This is an adaptation of a method of \cite{prev} based on a proper vertex coloring of the graph. Roughly speaking, in each stage $i$, all the vertices of color $i$ select a value for their random bits to ensure that the conditional expectation of some statistic of interest does not increase. The objective function here acts in a black-box way and can be almost completely arbitrary.

In Section~\ref{sec:derand-noprop}, we develop our second derandomization technique. This extends this first method to use a \emph{non-proper} vertex coloring, which may have fewer colors. This new algorithm is fundamentally white-box: it requires a carefully tailored pessimistic estimator for the conditional expectation. To state it somewhat informally, the estimator must be ``multilinear'' with respect to the coloring. This allows all the vertices of a color class to act simultaneously without non-linear interactions. 

This multilinearity condition is a significant restriction, but we show that there is a natural way to satisfy it for Chernoff bounds. To explain this at a very high level, concentration bounds with probabilities of order $\delta$ depend on the joint behavior of $w$-tuples of vertices for $w = O(\log \tfrac{1}{\delta})$. For an appropriately chosen  coloring, most such $w$-tuples will have vertices of different colors.

In Sections~\ref{rnd-sec2} and ~\ref{splitg2sec}, we turn this machinery to derandomize direct rounding. We slow down the random process: instead of selecting the edges with probability $p = h(e) \log r$ at once, we go through multiple stages in which each edge is selected with probability $1/2$. We then use the conditional expectations method at each stage, ensuring that the weight of the retained edges at the end is close to the expected weight initially. This is the most technical part of the paper. The statistic is a complex, non-linear function, so instead of directly computing its conditional expectation, we carefully construct a family of pessimistic estimators which approximate it but are amenable to the derandomization method in Section~\ref{sec:derand-noprop}. 

This algorithm runs in $\tilde O(r \log \Delta + \log^2 \Delta)$ rounds to generate the $O(r)$-approximate maximum weight matching.  We find it remarkable that the global statistic $\sum_{e \in L} a(e)$ can be derandomized in this completely local way, without dependence on $n$. (The overall algorithm requires an edge coloring of $H$, and obtaining this requires $O(\log^* n)$ rounds.)

In Section~\ref{find-match}, we describe the initial step of obtaining the fractional matching. This uses the LP algorithm of \cite{kmw} as well as a few additional quantization steps. For the randomized algorithm, we also randomly sparsify the original graph, reducing the degree from $\Delta$ to $\log \tfrac{1}{\delta}$ where $\delta$ is the desired failure probability, and then we execute the deterministic algorithm. Note that the randomized algorithm, which is based on the derandomization of direct rounding, has a  failure probability which is exponentially smaller than direct rounding itself.

In Section~\ref{graph-match}, we develop an algorithm for $(1+\epsilon)$-approximation to GMWM. As we have discussed, this algorithm repeatedly finds a collection of disjoint high-weight augmenting paths, which we represent in terms of a hypergraph matching. It is critical here that our algorithm finds a \emph{high-weight} hypergraph matching, not merely a maximal matching.  We also provide further details on lower bounds for GMWM.

  In Section~\ref{maximal-match}, we discuss HMM and applications to edge-list-coloring. The basic algorithm for HMM is simple: at each stage, we find an approximate maximum-cardinality matching in the residual hypergraph. Our HMWM approximation  ensures that the maximum matching cardinality in the residual graph decreases by a $1 - 1/O(r)$ factor in each stage, which leads to a maximal matching  $O(r \log n)$ repetitions.  The randomized algorithm also takes advantage of a variant of the ``shattering method'' of \cite{beps} (we provide a self-contained description of this process).
    
  In Section~\ref{edge-sec}, we describe a more elaborate application of HMM to approximate Nash-Williams decomposition. We describe both randomized and deterministic algorithms for this task. Counter-intuitively,  the deterministic algorithm is built on our \emph{randomized} HMM algorithm.
  
\subsection{Comparison with related work}
Our algorithmic framework can be interpreted combinatorially as the following problem: given a hypergraph $H = (V,E)$, find a matching of weight approximately $\frac{a(E)}{\Delta r}$ for some edge-weighting function $a$.  Let us summarize the basic approach of \cite{fgk} and \cite{prev} for this problem and how we improve the complexity. 

The algorithm of \cite{fgk}  is based on a primal-dual method: a vertex cover (which is the dual problem to maximum matching) is maintained to witness the optimality of the selected edge subset. This algorithm gives an $O(r^3)$-approximation to HMWM in $O(r^2 (\log \Delta)^{6 + \log r} + \log^* n)$ rounds.

The algorithm of \cite{prev}, like ours, is based on derandomized rounding. They only aim for a hypergraph maximal matching, not an approximate HMWM. The key algorithmic subroutine for this is \emph{degree-splitting}: namely, selecting an edge-set $E' \subseteq E$ which has degree at most $\tfrac{\Delta}{2} ( 1 + \epsilon)$ and which contains approximately half of the edges. This has a trivial $0$-round randomized algorithm, by selecting edges independently with probability $1/2$. To derandomize this, \cite{prev} divides $H$ into ``virtual nodes'' of degree $\frac{\log n}{\epsilon^2}$, and then uses a proper vertex coloring of the resulting line graph (which has maximum degree $\frac{r \log n}{\epsilon^2})$.

To reduce the degree further, \cite{prev} repeats this edge-splitting process for $s \approx \log_2 \Delta$ steps. This generates nested edge-sets $E = E_0 \supseteq E_1 \supseteq \dots \supseteq E_s$ wherein each $E_i$ has degree at most $\Delta ( \frac{1+\epsilon}{2})^i$ and has $|E_i| \approx 2^{-i} |E|$. The final set $E_s$ has very small maximum degree, and a simpler algorithm can then be used to select a large matching of it.

The process of generating edge sets $E_0, \dots, E_s$ with decreasing maximum degree is quite similar to our derandomization of direct rounding. Both algorithms generate nested edge sets which simulate the random process of retaining edges independently. But the key difference is that \cite{prev}  ensures that \emph{all} the vertices have degree at most $\frac{\Delta}{2}(1+\epsilon)$ \emph{in every stage}. Since they use a union bound over the vertices, they must pay a factor of $\log n$ in the run-time.  Furthermore, in order to use degree-splitting over $s$ stages with constant-factor total loss, they need to satisfy $(1+\epsilon)^s = O(1)$, i.e., $\epsilon \approx 1/\log \Delta$. These strict concentration bounds give a complexity of $\tilde O(\log \Delta \times r \log n/\epsilon^2) = \tilde O(r \log n \log^3 \Delta)$ over all stages.

Let us  discuss how our algorithm avoids these two issues. First, to avoid dependence on $n$, we do not insist that all vertices have their degree reduced. We discard the vertices for which certain bad-events occur, e.g. the degree is much larger than expected. These are rare so this does not lose too much in the weight of the matching.

Second, observe that it would not be unusual in the random process for the degree of a vertex $v$ to deviate significantly from its mean value in a \emph{single} stage. Thus, we only keep track of the \emph{total} deviation of $\deg(v)$ from its mean value, aggregated over all vertices $v$ and stages $i$. We achieve this through a carefully crafted potential function to analyze direct rounding.  By allowing more leeway for each vertex per stage, we get away with looser concentration bounds.

\subsection{Notation and conventions}
For a graph $G = (V,E)$ and $v \in V$, we define $N[v]$ to be the inclusive neighborhood of vertex $v$, i.e. $\{v \} \cup \{ w \mid (w, v) \in E \}$. For a hypergraph $H = (V,E)$ and $v \in V$, we define $N(v)$ to be the set of edges containing $v$. We define $\text{deg}(v) = |N(v)|$ and for $L \subseteq E$ we define $\text{deg}_L(v) = |N(v) \cap L|$. Unless stated otherwise, $E$ may be a multi-set. 

For a set $X$, we define $2^X$ to be the power set of $X$, i.e. the collection of subsets $Y \subseteq X$. For integer $k$, we define $\binom{X}{k} \subseteq 2^X$ to be the collection of $k$-element subsets of $X$.

For a graph $G = (V,E)$, we define the power graph $G^t$ to be thee graph on vertex set $V$, with an edge $(u,v)$ if there is a path of length up to $t$ in $G$ from $u$ to $v$. Note that if $G$ has maximum degree $\Delta$, then $G^t$ has maximum degree at most $\Delta^t$.

We define a \emph{fractional matching} to be a function $h: E \rightarrow [0,1]$ such that $\sum_{e \in N(v)} h(e) \leq 1$ for every $v \in V$. This should be distinguished from a fractional solution to the matching polytope of a graph, which also includes constraints for all odd cuts. 

We define an \emph{edge-weighting} to be a function $a: E \rightarrow [0,\infty)$. For an edge-weighting $a$ and an edge subset $L \subseteq E$, we define $a(L) = \sum_{e \in L} a(e)$. Similarly, for a fractional matching, we define $a(h) = \sum_{e \in E} h(e) a(e)$. For a hypergraph $H = (V,E)$ we write $a(H)$ as shorthand for $a(E)$ and we define $a^*(H)$ to be the maximum possible value of $a(h)$ over all fractional matchings $h$ of $H$.  
  
  For any boolean predicate $\mathcal P$, we use the Iverson notation so that $[[\mathcal P]]$ is the indicator function that $\mathcal P$ is true, i.e. $[[\mathcal P]] = 1$ if $\mathcal P$ is true and $[[\mathcal P]] = 0$  otherwise.

    \subsection{The LOCAL model}
    \label{model:sec}
    Our algorithms are all based on the LOCAL model for distributed computations in a hypergraph. This is a close relative to Linial's classic LOCAL graph model \cite{lin92, pel00}, and in fact the main motivation for studying hypergraph LOCAL algorithms is because they are useful subroutines for LOCAL graph algorithms.

 The LOCAL model for graphs has two variants depending on the role of randomness. In the deterministic variant, each vertex is provided with a unique ID which is a bit-string of length $\Theta(\log n)$; here $n$ is a global parameter passed to the algorithm. A vertex has a list of the ID's of its neighbors. In each round a vertex can perform arbitrary computations and transmit messages of arbitrary sizes to its neighbors. At the end of this process, each vertex must make a decision for a graph problem. For example, if the graph problem is to compute a maximal independent set, then each vertex $v$ sets a flag $F_v$ indicating whether it has joined the MIS.

In the randomized LOCAL graph model, each vertex maintains a private random bit-string $R_v$ drawn from some distribution $\mu$.  We define $\vec R \in \mu^V$ to be the overall collection of values $R_v$. All steps except the generation of $\vec R$ can be regarded as deterministic, and so each $F_v$ can be viewed as a function on the domain $\mu^V$. At the end of the process, the flags $F_v$ must correctly solve the graph problem w.h.p., i.e. with probability at least $1 - n^{-c}$ for any desired constant $c > 0$.

To define the LOCAL model on a hypergraph $H$, we first form the \emph{incidence graph} $G = \text{Inc}(H)$; this is a bipartite graph in which each edge and vertex of $H$ corresponds to a vertex of $G$. If $u_v$ and $u_e$ are the vertices in $G$ corresponding to the vertex $v \in V$ and $e \in E$, then $G$ has an edge $(u_e, u_v)$ whenever $v \in e$. The LOCAL model for hypergraph $H$ is simply the LOCAL graph model on $G$. In other words, in a single timestep on the hypergraph $H$, each vertex can send arbitrary messages to every edge $e \in N(v)$, and vice versa.

Note that the length-$\ell$ paths of a graph $G$ can be represented as hyperedges in an auxiliary rank-$\ell$ hypergraph $H$. A single round of the LOCAL model on $H$ can be simulated in $O(\ell)$ rounds of the LOCAL graph model on $G$. 

Graph and hypergraph algorithms may depend on statistics such as the \emph{maximum degree} $\Delta$, \emph{rank} $r$, and vertex count $n$. These parameters cannot be computed locally. As is standard in distributed algorithms, we consider $\Delta, r, n$ to be globally-known upper-bound parameters with the guarantee that $|N(v)| \leq \Delta, |e| \leq r, |V| \leq n$ for every vertex $v$ and edge $e$. When $V$ is understood, we may say that $E$ has maximum degree $\Delta$ and maximum rank $r$. We also assume throughout that $r \geq 2$, as the cases $r = 0$ and $r = 1$ are trivial. 

\section{Derandomization via proper vertex coloring}
\label{rnd-sec}
We begin with a general method of converting randomized LOCAL algorithms into deterministic ones through proper vertex colorings and conditional expectations. This is a slight generalization of \cite{prev};  we describe it here to set the notation and since some of the parameters are slightly different.

Consider a 1-round randomized algorithm $A$ run on a graph $G = (V,E)$. Each vertex $v$ draws its random string $R_v$, runs algorithm $A$, and outputs a real-valued flag $F_v$. Since the algorithm $A$ takes just one round, each value $F_v(\vec R)$ is determined by the values $R_w$ for $w \in N[v]$. We emphasize that the flag $F_v$ here is not necessarily an indicator that the overall algorithm has failed with respect to $v$, and the underlying graph problem may not even be locally checkable.

\begin{lemma}
  \label{basicderand}
  Suppose that $G^2$ has maximum degree $d$.  Then there is a deterministic  algorithm in $O(d + \log^* n)$ rounds to determine values $\vec \rho$ for the random bits $\vec R$, such that
  when (deterministically) running $A$ with the values $\vec R = \vec \rho$, it satisfies
  $$
  \sum_v F_v(\vec \rho) \leq \sum_v \bE[ F_v(\vec R) ]
  $$

  Furthermore, when given a $O(d)$ coloring of $G^{2}$ as input, the $\log^* n$ term can be omitted.
\end{lemma}
\begin{proof}
See \cite{prev} for a full exposition; we provide just a sketch here. 

If we are not already given  a proper vertex coloring of $G^{2}$ with $O(d)$ colors,  we  use the algorithm of \cite{fraigniaud} to obtain this in $\tilde O(\sqrt{d}) + O(\log^* n)$ rounds. Next, we proceed sequentially through the color classes; at the $i^{\text{th}}$ stage, every vertex $v$ of color $i$ selects a value $\rho_v$ to ensure that the expectation of $F_v + \sum_{(u,v) \in E} F_u$, conditioned on $R_v = \rho_v$, does not increase. Note that all the vertices of color $i$ are non-neighbors, so they do not interfere during this process.
\end{proof}

\Cref{basicderand} is stated in terms of \emph{minimizing} the sum $\sum_{v}  F_v(\vec \rho)$; by replacing $F$ with $-F$, we can also maximize the sum, i.e. get $\sum_{v}  F_v(\vec \rho) \geq \bE[ \sum_{v} F_v(\vec R)]$. In our applications, we will use whichever form (maximization or minimization) is most convenient; we do not explicitly convert between these two forms by negating the objective functions.

To derandomize hypergraph  algorithms, we apply \Cref{basicderand} to the incidence graph $G = \text{Inc}(H)$. It is convenient to rephrase Lemma~\ref{basicderand} in terms of $H$ without explicit reference to $G$.
\begin{definition}
  \label{def-good-coloring}
  For a hypergraph $H$ of rank $r$ and maximum degree $\Delta$, and incidence graph $G$, we define a \emph{good coloring of $H$} to be a proper vertex coloring of $G^2$ with $\poly(r, \Delta)$ colors.
\end{definition}

\begin{lemma}
  \label{basicderand2}
  Let $A$ be a randomized $1$-round algorithm run on a hypergraph $H = (V,E)$.  Each $u \in V \cup E$ has private random bit-string $R_u$ and at its termination $A$ outputs quantities $F_u$.  If we are provided a good coloring of $H$, then there is a deterministic $O(r \Delta)$-round algorithm to determine values $\vec \rho$ such that 
  when (deterministically) running $A$ with the values $\vec R = \vec \rho$, it satisfies
  $$
  \sum_{u \in V \cup E}  F_u(\vec \rho) \leq \bE[ \sum_{u \in V \cup E} F_u(\vec R)]
  $$
\end{lemma}
\begin{proof}
  Let $G$ be the incidence graph of $H$. Note that $G^2$ has maximum degree $d = \Delta r$. The good coloring of $H$ is a $k$-coloring of $G^2$ with $k = \poly(d)$; we transform this into an $O(d)$-coloring of $G^2$ using the algorithm \cite{fraigniaud} in $o(d) + O( \log^* k)$ rounds. Then apply Lemma~\ref{basicderand} with respect to $G$.
\end{proof}

A typical strategy for our algorithms will be to first get a good coloring of the original input hypergraph $H$, in $O(\log^* n)$ rounds, using the algorithm of \cite{lin92}. Whenever we modify $H$ (by splitting vertices, replicating edges, etc.) we will update this coloring. All these subsequent updates can be performed in $O(\log^*(r \Delta))$ rounds, which will be negligible compared to the runtime of the overall algorithm.  These coloring updates are routine but quite cumbersome to describe. For simplicity of exposition, we will mostly ignore them for the remainder of the paper.

As a simple application of Lemma~\ref{basicderand2}, which we need later in our algorithm, we consider a randomized procedure to find a hypergraph matching.
\begin{lemma}
  \label{match-prop2} 
Let $H = (V,E)$ be a hypergraph with a good coloring and an edge-weighting $a$. There is a deterministic $O(r \Delta)$-round algorithm to compute a matching $M$ with $a(M) \geq \Omega( \frac{ a(H)}{r \Delta} )$.
\end{lemma}
\begin{proof}
  Consider the following 1-round randomized algorithm: we first form an edge-set $L$, wherein each edge $e \in E$ goes into $L$ independently with probability $p = \frac{0.1}{r \Delta}$, and we then form a matching $M$ from $L$ by discarding any pair of intersecting edges. 
  
  It is straightforward to see that the resulting matching $M$ has expected weight $\Omega( \frac{a(H)}{r \Delta})$.  We derandomize this via a pessimistic estimator of $\bE[a(M)]$. Let us define the following flags computed by each edge $e$ and vertex $v$ of $H$:
    $$
  F_e = [[e \in L]] a(e), \qquad \qquad  F_v = \sum_{\substack{e,  e' \in N(v) \\ e \neq e'}} -[[e \in L \wedge e' \in L]] a(e)
$$

We compute $\sum_{u \in V \cup E} \bE[F_u]$ as:
  $$
  \sum_{u \in V \cup E} \bE[F_u] = \sum_{e \in E} a(e) \Pr(e \in L) - \sum_{v \in V} \sum_{\substack{e, e' \in N(v) \\ e' \neq e}} a(e) \Pr( e \in L \wedge e' \in L)
  $$
  and we can estimate
  \begin{align*}
&  \sum_{v \in V} \sum_{\substack{e, e' \in N(v) \\ e \neq e'}} a(e) \Pr( e \in L \wedge e' \in L) \leq  \sum_{v \in V} \sum_{e \in N(v)} a(e) p^2 \deg(v) \leq \frac{0.01}{r^2 \Delta} \sum_{v \in V} \sum_{e \in E} a(e) \leq \frac{0.01 a(E)}{r \Delta}
    \end{align*}
  where the last inequality holds by double-counting, noting that the rank of $H$ is at most $r$.

  This implies that $$
  \sum_{u \in V \cup E} \bE[F_u] = \frac{0.1 a(E)}{r \Delta}  - \frac{ 0.01 a(E)}{r \Delta} \geq \Omega( \frac{a(E)}{r \Delta} )
  $$ 

  Lemma~\ref{basicderand2} gives an $O(r \Delta)$-round deterministic algorithm to find random values $R_u$ with $\sum_u F_u \geq \sum_u \bE[F_u]$. Let $L$ denote the corresponding set of marked edges. When we form the matching $M$ from $L$ by discarding edges, we get
  $$
  a(M) \geq \sum_{e \in E} [[e \in L]] a(e) - \sum_{\substack{e,  e' \in N(v) \\ e \neq e'}} [[e \in L \wedge e' \in L]] a(e) = \sum_u F_u
  $$
and we have already seen this is $\Omega( \frac{a(H)}{r \Delta} )$.
\end{proof}

The main limitation of Lemma~\ref{basicderand2}, used on its own, is its polynomial dependence on $\Delta$. For some problems, following \cite{prev}, this can be circumvented by splitting vertices into smaller ``virtual nodes.'' In this way, Lemma~\ref{basicderand2} can be used to obtain edge-colorings in $\polylog \Delta$ rounds. The following definitions are useful to characterize this process:
\begin{definition}[defective and balanced edge colorings]
For a hypergraph $H$ and map $\chi: E \rightarrow \{1, \dots, k \}$, we say $\chi$ is a \emph{$k$-edge-coloring}. We say that $\chi$ is \emph{$t$-defective} if every vertex $v$ has at most $t$ edges of any given color $j$; more formally, if $|N(v) \cap \chi^{-1}(j)| \leq t$ for all $v \in V, j \in \{1, \dots, k \}$. 

We say that a $k$-edge-coloring $\chi$ is \emph{balanced} if it is $t$-defective for $t = O(\Delta/k)$.
\end{definition}
Our more advanced derandomization methods will need balanced edge-colorings. Note that a proper edge coloring corresponds, in our terminology, to a $1$-defective edge coloring; this would use $k = r \Delta$ colors, and hence not be balanced.  The trivial randomized coloring algorithm gives a balanced edge coloring for $t = \log n$. One main contribution of \cite{prev} is to derandomize this, getting a balanced coloring with $t = \polylog n$. However, this dependence on $n$ is not suitable for us; we want $t$ to depend only on local parameters $r, \Delta$. 

A straightforward application of the Lov\'{a}sz Local Lemma (LLL) shows that a balanced edge-coloring exists for $t = \log (r \Delta)$. By iterating the LLL, the degree $\Delta$ can be successively reduced, yielding a balanced edge-coloring with $t = \log r$.  Unfortunately, even the randomized LLL algorithms are too slow for us. We must settle for something slightly weaker, namely, a \emph{partial} defective edge-coloring. This has a simple randomized algorithm which serves as a ``poor man's LLL'': randomly $k$-color the edges and discard vertices with more than $c \Delta/k$ edges of any color. 

 We obtain the following result; the proof is very similar to \cite{prev} and is deferred to Appendix~\ref{degree-split-app}.
\begin{lemma}
  \label{lem-split-crude}
  There is an absolute constant $C > 0$ for which the following holds.
  
Suppose $H = (V,E)$ with a good coloring and an edge-weighting $a$. For any $\delta \in (0,\tfrac{1}{2})$ and integer $k \geq 2$ satisfying $\Delta \geq C k \log \tfrac{r \log k}{\delta}$, there is an $\tilde O(r \log \tfrac{1}{\delta} \log^3 k)$-round deterministic algorithm to find an edge set $E' \subseteq E$ and $k$-edge-coloring $\chi$ of $E'$ such that $a(E') \geq (1-\delta) a(E)$ and $\chi$ is $(4 \Delta/k)$-defective on $E'$.
\end{lemma}
In particular, note that the coloring $\chi$ is balanced.

\section{Derandomization without proper vertex coloring}
\label{sec:derand-noprop}
The use of a proper vertex coloring $\chi$ in Lemma~\ref{basicderand}  is somewhat limiting. In this section, we relax this condition. To explain briefly: suppose we try to use the conditional expectation method of Lemma~\ref{basicderand} with a non-proper coloring $\chi$. In this case, vertices with the same color may interfere with each other, so their contributions would be computed incorrectly.  To avoid these errors, we can carefully construct the statistic $\Phi= \sum_u F_u$ to be ``multilinear'' with respect to the coloring $\chi$. This avoids problematic non-linear interactions between vertices. (Linear interactions do not cause problems.)  This is quite different from \Cref{basicderand}, in which the function $F_u$ is almost arbitrary and is treated in a black-box way.

The derandomization result we get this way may seem very abstract. We follow with an example showing how it applies to degree splitting and concentration bounds for sums of random variables.

Note that \cite{ks18} uses a similar derandomization strategy based on a non-proper vertex coloring, but the error analysis is very different. The algorithm of \cite{ks18} \emph{ignores} all interactions between vertices with the same color, followed by a postprocessing step to correct the resulting errors. By contrast, our statistic $\Phi$ is an approximation to the true statistic of interest (incurring some error),  but the algorithm handles it ``correctly'' and incurs no further error in derandomizing it.

\subsection{The derandomization lemma}
Consider a graph $G = (V,E)$ and some global statistic $\Phi: \mu^V \rightarrow \mathbb R$ which is a function of the random bits $R_v$. Each vertex may have some additional, locally-held state; for example, we may be in the middle of a larger multi-round algorithm and so each vertex will have  information about its $t$-hop neighborhood. The function $\Phi$ may depend on this vertex state as well; to avoid burdening the notation, we do not explicitly write the dependence.

The required multilinearity condition is defined in terms of the directional derivative structure of the function $\Phi$. Formally, for any vertex $v \in V$ and any pair of values $u, u' \in \mu$, the \emph{derivative} is the function $D_{v,u,u'} \Phi: \mu^{V - \{v \}} \rightarrow \mathbb R$ defined as follows:
\begin{align*}
  (& D_{v, u,u'} \Phi) (x_1, \dots,  x_{v-1}, x_{v+1}, \dots, x_n ) \\
  & \qquad \qquad \qquad = \Phi( x_1, \dots, x_{v-1}, u , x_{v+1}, \dots, x_n) - \Phi( x_1, \dots, x_{v-1}, u', x_{v+1}, \dots, x_n)
\end{align*}
Observe that $D_{v,u,u} \Phi = 0$ and $D_{v,u,u'} \Phi = -D_{v,u',u}$.

We say that the function $\Phi$ is \emph{ uncorrelated} for vertices $v, v'$ if for any values $u, u'$, the function $D_{v,u,u'} \Phi (x_1, \dots, x_n)$ does not depend on the value of $x_{v'}$. 

Note that $D_{v,u,u'} \Phi$ is itself a function on a smaller vertex set $V - \{v \}$, so we can talk about its derivatives as well. Thus, an equivalent condition for $\Phi$ to be uncorrelated for vertices $v_1, v_2$ is that for any values $u_1, u_1', u_2, u_2'$, the function $D_{v_1, u_1, u_1'} D_{v_2, u_2, u_2'} \Phi$ is identically zero.

\textbf{Remark: binary-valued probability spaces.}
In many applications,  the underlying probability space $\mu$ is over the set $\{0, 1 \}$. We then write $D_{v}$ as shorthand for $D_{v, 1,0}$. Note that the other directional derivatives for $v$ can be obtained from $D_{v,1,0} \Phi$, since $D_{v,0,0} \Phi = D_{v,1,1} \Phi = 0$ and $D_{v,0,1} \Phi = -D_{v,1,0} \Phi$.   Also, $v, w$ are  uncorrelated iff $D_v D_w \Phi$ is identically zero.
 
We are now ready to state our main derandomization lemma.
\begin{lemma}
\label{basicderand3}
Suppose we have a vertex coloring $\chi: V \rightarrow \{1, \dots, k \}$ (not necessarily proper) for $G$, where the potential function $\Phi: \mu^V \rightarrow \mathbb R$ has the following properties:
\begin{enumerate}
\item[(A1)] For any distinct vertices $v, w$ with $\chi(v) = \chi(w)$, the function $\Phi$ is  uncorrelated for $v, w$.
\item[(A2)] For any vertex $v$ and values $u, u'$,  the function $D_{v, u, u'} \Phi (x_1, \dots, x_n)$ can be locally computed by $v$ given the values of $x_w$ for $w \in N[v]$.
\end{enumerate}

Then there is a deterministic $O(k)$-round algorithm to determine values $\vec \rho$ for the random bits,  such that $\Phi(\vec \rho) \leq \bE[ \Phi (\vec R) ]$.
\end{lemma}
\begin{proof}
  We proceed through stages $i = 1, \dots, k$; at each stage $i$, every vertex  $v$ with $\chi(v) = i$ selects some value $\rho_v$ to minimize $\bE[ \Phi \mid R_v = \rho_v]$, and permanently commits to $R_v = \rho_v$.

 To analyze this process, suppose that we are at stage $i$ and we have fixed the values $\rho_w$ for every vertex $w$ with $\chi(w) < i$. All expectation calculations will be conditioned on the values of $R_w$ for such vertices $w$. 
  
  We first claim each vertex $v$ can determine the value $\rho_v$ in $O(1)$ rounds. For, observe that $u$ minimizes the conditional expectation $\bE[ \Phi \mid R_v = u]$ if and only if $\bE[ D_{v, u, u'} \Phi] \leq 0$ for all values $u'$. In $O(1)$ rounds $v$ can query the values of $\rho_w$ for $w \in N[v]$. By Property (A2), this is enough to determine $\bE[ D_{v, u, u'} \Phi ]$ for all $u, u'$. This allows $v$ to select $\rho_v$ appropriately.
  
  Next, we claim that the conditional expectation of $\Phi$ does not increase during stage $i$. Suppose that the color-$i$ vertices are $v_1, \dots, v_s$ and they select values $\rho_{v_1}, \dots, \rho_{v_s} $ respectively. Note that, by Property (A1), the value of $D_{v_j,u, u'} \Phi$ does not depend on the values $R_{v_1}, \dots, R_{v_{j-1}}$.

  We claim now that for all $j = 1, \dots, s$ we have
  \begin{equation}
  \label{bbff1}
  \bE[ \Phi \mid R_{v_1} = \rho_{v_1}, \dots, R_{v_j} = \rho_{v_j}] \leq \bE[ \Phi \mid R_{v_1} = \rho_{v_1}, \dots, R_{v_{j-1}} = \rho_{v_{j-1}}]
  \end{equation}
  To show this, we  calculate:
  \begin{align*}
  &  \bE[ \Phi \mid R_{v_1} = \rho_{v_1}, \dots, R_{v_{j-1}} = \rho_{v_{j-1}}] = \sum_{u'} \bE[ \Phi \mid R_{v_1} = \rho_{v_1}, \dots, R_{v_{j-1}} = \rho_{v_{j-1}}, R_{v_j} = u] \mu(u') \\
    & \qquad = \bE[ \Phi \mid R_{v_1} = \rho_{v_1}, \dots, R_{v_{j-1}} = \rho_{v_{j-1}},  R_{v_j} = \rho_{v_j}] \\
    & \qquad \qquad - \sum_{u'} \bE[ D_{v,\rho_{v_j}, u'} \Phi \mid R_{v_1} = \rho_{v_1}, \dots, R_{v_{j-1}} = \rho_{v_{j-1}}] \mu(u') \\
    & \qquad = \bE[ \Phi \mid R_{v_1} = \rho_{v_1}, \dots, R_{v_{j-1}} = \rho_{v_{j-1}},  R_{v_j} = \rho_{v_j}] - \sum_{u'} \bE[ D_{v,\rho_{v_j}, u'} \Phi] \mu(u') \qquad \text{Property (A1)} \\
      & \qquad \geq \bE[ \Phi \mid R_{v_1} = \rho_{v_1}, \dots, R_{v_{j-1}} = \rho_{v_{j-1}},  R_{v_j} = \rho_{v_j}]  \qquad \text{by our choice of $\rho_v$}
  \end{align*}

  From Eq.~(\ref{bbff1}) and induction we have $\bE[ \Phi \mid R_{v_1} = \rho_{v_1}, \dots, R_{v_s} = \rho_{v_s}] \leq \bE[ \Phi ]$, as desired.
\end{proof}

Lemma~\ref{basicderand} can be viewed a special case of Lemma~\ref{basicderand3}: for, consider a $1$-round LOCAL algorithm which computes a function $F_v$ for each vertex $v$. The potential function $\Phi = \sum_{v \in V} F_v$ satisfies Lemma~\ref{basicderand3} with respect to the graph $G^{2}$ and a proper vertex coloring $\chi$ of $G^2$. In particular, condition (A1) is satisfied since any vertices $v, w$ of the same color must have distance at least $2$ in $G$ and so $v,w$ are uncorrelated for $\Phi$.

\subsection{Example: derandomizing edge-splitting}
\label{toy-sec} For better motivation, consider the following random process. For a hypergraph $H = (V,E)$, each edge $e \in E$  is chosen for a set $L$ with probability $1/2$. We want to approximately split each vertex; formally, we say a vertex $v$ is \emph{bad} if either $\deg_L(v) \geq t$ or $\deg_{E-L}(v) \geq t$, for some threshold value $t = \tfrac{\Delta}{2} (1+\epsilon)$. We would like to select $L$ so that there are no bad vertices. We assume that $\Delta \gg \frac{\log n}{\epsilon^2}$, so that randomly-chosen $L$ has  this property with high probability.

We will use a powerful result of \cite{chernoff-hoeffding} relating Chernoff bounds to symmetric polynomials. If we define the indicator variables $X_e = [[ e \in L ]]$, we have the following inequality for any integer $w \leq t$:
$$
[[ \deg_L(v) \geq t ]] \leq \frac{ \binom{ \deg_L(v)}{w}  }{ \binom{ t}{w}} = \frac{ \sum_{ W \in \binom{N(v)}{w}} \prod_{e \in W} X_e} { \binom{t}{w}}
$$
Taking expectations, this gives:
\begin{equation}
  \label{cbeqn1}
\Pr( \deg_L(v)  \geq t ) \leq \beta \sum_{W \in \binom{N(v)}{w}} \bE[\prod_{e \in W} X_e] = \beta \tbinom{\deg(v)}{w} 2^{-w}
\end{equation}
where $\beta = 1/\binom{t}{w}$ is a normalization term. 
Further calculations of \cite{chernoff-hoeffding} show that for $w = \lceil t - \mu \rceil$ the RHS is at most the well-known Chernoff bound $( \frac{e^{\epsilon}}{(1+\epsilon)^{1+\epsilon}})^{\mu}$ for $\mu = \tfrac{\Delta}{2}$. Smaller values of $w$ also yield powerful concentration bounds, with probability bounds roughly inverse exponential in $w$.

Thus, as a proxy for the number of bad vertices, a natural pessimistic estimator is given by:
$$
\Phi(L) =  \beta \sum_v \tbinom{ \deg_L(v) }{w} + \tbinom{\deg_{E-L}(v)}{w} = \beta \sum_v \sum_{W \in \binom{N(v)}{w}}\Bigl( \prod_{e \in W} X_e + \prod_{e \in W} (1 - X_e) \Bigr)
$$

If any vertex $v$ is bad, then this implies that $\Phi(L) \geq 1$. When $\Delta \gg \frac{\log n}{\epsilon^2}$ and $w = \Theta( \frac{\log n}{\epsilon})$, one can check that $\bE[\Phi(L)] \leq \frac{1}{\poly(n)}$. Thus, with these parameters, Markov's inequality applied to $\Phi(L)$ shows that with high probability there are no bad vertices.

To derandomize this, we will apply \Cref{basicderand3} to the graph $G = \text{Inc}(H)^2$. We only care about the nodes of $G$ corresponding to edges of $H$; the nodes of $G$ corresponding to vertices are immaterial and will be ignored. Thus, the vertex-coloring of $G$ used for \Cref{basicderand3} corresponds to an edge-coloring of $H$. To avoid confusion, we will phrase our algorithm directly in terms of the edge-coloring of $H$, without explicitly translating back into the underlying graph $G$.

Let us now suppose that we are given such a balanced $k$-edge coloring $\chi$ of $E$, which has defectiveness $d = O(\Delta/k)$; the choice of $k$ and the role played by $\chi$ will become clear shortly. We want to apply Lemma~\ref{basicderand3} to the graph $G$ and coloring $\chi$ to derandomize statistic $\Phi$, but, unfortunately, $\Phi$ does not satisfy the required property (A1). The problem is that some set $W$ may contain multiple edges with the same color; for a pair of edges $e_1, e_2 \in N(v)$, a term such as $\prod_{e \in W} X_e + \prod_{e \in W} (1 - X_e)$ will have a non-vanishing second derivative $D_{e_1} D_{e_2} \prod_{e \in W} X_e$.

To avoid this, we need another statistic $\Phi'$ which approximates $\Phi$. For any vertex $v$, define $\mathcal U_v$ to be the set of all subsets $W \in \binom{N(v)}{w}$  such that all edges in $W$ have distinct colors. If we restrict the sum to only the subsets $W \in \mathcal U_v$, then we get a statistic which is compatible with \Cref{basicderand3}:
$$
\Phi'(L) = \beta \sum_v | \mathcal U_v \cap 2^L | + |\mathcal U_v \cap 2^{E - L}| = \beta \sum_v \sum_{W \in \mathcal U_v} \Bigl(  \prod_{e \in W} X_e + \prod_{e \in W} (1 - X_e) \Bigr)
$$
  
To show that $\Phi'$ satisfies property (A1), suppose that $\chi(e_1) = \chi(e_2)$. Then any $W \in \mathcal U_v$ contains at most one of $e_1, e_2$, and so the second derivative $D_{e_1} D_{e_2} \beta \bigl( \prod_{e \in W} X_e + \prod_{e \in W} (1 - X_e) \bigr)$ is zero. To show that it satisfies (A2), we compute the first derivative $D_{e_1}$ as:
$$
D_{e_1} \Phi' =  \beta \sum_v \sum_{W \in \mathcal U_v: e_1 \in W} D_{e_1}  \bigl( \prod_{e \in W} X_e + \prod_{e \in W} (1 - X_e) \bigr)
$$
For such $W$, all the other edges $e \in W$ also involve vertex $v$, and so $e$ is a neighbor to $e_1$ in $G^2$. So all such terms can be computed locally by $e_1$.

Applying Lemma~\ref{basicderand3} to $\Phi'$ gives a subset $L$ with $\Phi'(L) \leq \bE[\Phi'(L)]$, and clearly $\bE[\Phi'(L)] \leq \bE[\Phi(L)]$. We next need to argue that the resulting set $L$ has no bad vertices. This step is where we need to keep track of some error terms or slack between $\Phi'$ and $\Phi$.

Suppose that $\deg_L(v) \geq t$ for some vertex $v$. To count $|\mathcal U_v|$, note that we have at least $t$ choices for the first edge $f_1$ in a set $W$ in $\mathcal U_v$. Since $\chi$ has defectiveness $d$, we have at least $t - d$ choices for the second edge $f_2$ of $W$ (which must have a different color than $f_1$), and similarly $t - 2 d$ choices for $f_3$, and so on. Continuing this way, we see that
$$
| \mathcal U_v | \geq \frac{t (t-d)(t - 2d) \dots (t - (w-1) d) }{w!} \geq \binom{t}{w} \bigl(\frac{t - w d}{t} \bigr)^w
$$

For our parameters, we can see that $\bigl(\frac{t - w d}{t} \bigr)^w \geq e^{-O(w^2 d/t)} \geq e^{-O(w^2/k)}$. Hence we have
$$
\Phi'(L) \geq \beta \tbinom{t}{w} e^{-O(w^2/k)}  = e^{-O(w^2/k)}
$$

Our derandomization will ensure that $\Phi'(L) \leq \bE[\Phi'(L)] \leq \bE[\Phi(L)]$. Thus, in order to guarantee there are no bad vertices, we need $\bE[\Phi(L)] \leq e^{-O(w^2/k)}$. Recalling that we have chosen the parameters so that $\bE[\Phi(L)] = \frac{1}{\poly(n)}$ and $w = \Theta(\frac{\log n}{\epsilon^2})$, we thus need to take
$$
k \geq \frac{ -\log \bE[\Phi(L)] }{w^2} = \Omega( \frac{\log n}{\epsilon^2} )
$$

If we can obtain a balanced coloring with this value $k$, then applying Lemma~\ref{basicderand3} takes $O(k) = O(\frac{\log n}{\epsilon^2})$ rounds. Thus, this gives an $O( \frac{\log n}{\epsilon^2})$-round algorithm for degree-splitting. By contrast, the method of \cite{prev} for degree-splitting based on ``virtual nodes''  would require $O( \frac{r \log n}{\epsilon^2})$ rounds. (See Appendix~\ref{degree-split-app} for a further discussion of the virtual node algorithm.) 

This toy example illustrates three important caveats in using Lemma~\ref{basicderand3}. First, we may have some natural statistic to optimize in our randomized algorithm (e.g.  number of bad vertices, total weight of edges retained in a coloring, indicator functions for whether a bad-event has occurred). This statistic typically will not satisfy condition (A1) directly. Instead, we must carefully construct a pessimistic estimator which does satisfy condition (A1).

Second we must allow some slack in our potential function to cover the errors from ignoring the non-linear interactions.  Here, for instance, we need to ensure that $\bE[\Phi(L)]$ is significantly below $1$. By contrast, if we directly used a conditional-expectations method, we would only need to ensure that $\bE[\Phi(L)] < 1$. 

Finally, \Cref{basicderand3} requires an appropriate coloring $\chi$, and the statistic $\Phi'$ will be defined in terms of $\chi$. In our hypergraph matching application, $\chi$ will be a balanced edge coloring. We first obtain this coloring using the derandomization Lemma~\ref{lem-split-crude}. This will also incur some small loss in the weight of the edges, which we have not discussed in this example. Each application of \Cref{basicderand3} is essentially a derandomization within a derandomization: we first use a relatively crude method to obtain $\chi$, and then use it to obtain a more refined bound via \Cref{basicderand3}. 

Our application next in Section~\ref{rnd-sec2} will be similar to this degree-splitting example, but much more  complex. The potential function will place different weights on certain vertices and edges.  In addition, instead of a fixed threshold value $t = \tfrac{\Delta}{2} (1+\epsilon)$ for each vertex $v$, the (effective) threshold value will also depend on the current degree of $v$.

\section{Derandomization of direct rounding}
\label{rnd-sec2}
We now use our derandomization methods to round a fractional matching $h$. As we have discussed, it is convenient to view this combinatorially, without explicit   reference to $h$. We show the following main  result:
\begin{theorem}
  \label{split2}
Let $H = (V,E)$ be a hypergraph of maximum degree $\Delta$ and rank $r$,  with a good coloring and an edge-weighting $a$. There is a $\tilde O(\log^2 \Delta + r \log \Delta)$-round deterministic algorithm to find a matching $M$ with
    $$
a(M) \geq \Omega\left( \frac{a(H)}{r \Delta} \right)
$$
\end{theorem}

This value of $a(M)$ is precisely what we would obtain by applying Lemma~\ref{match-prop2} to $H$. Unfortunately, Lemma~\ref{match-prop2} would take $O(r \Delta)$ rounds, which is much too large. 

Consider the following random process to reduce the degree of $H$: each edge $e \in E$ goes into an edge-set $J$ independently with probability $p = x / \Delta$, where $x = \polylog(r, \Delta)$. If any vertex $v$ has $\deg_J(v) \geq \Delta'$, where $\Delta'$ is some chosen threshold value with $\Delta' = \Theta(x)$, we discard all its neighboring edges. The remaining edge-set $J'$ then has maximum degree $\Delta'$ and a simple second-moment calculation shows that $\bE[ \frac{a(J')}{\Delta'} ] \geq \Omega( \frac{a(H)}{\Delta} )$. (See Proposition~\ref{maxmatch1} for further details).

Our main goal is to derandomize this to get $ \frac{a(J')}{\Delta'}  \geq \Omega( \frac{a(H)}{\Delta} )$ (in actuality, not expectation).  Instead of selecting edge-set $J$ in a single stage, we go through $s = \log_2(1/p)$ stages where each edge is retained with probability $1/2$. We let $J_0, \dots, J_s$ denote the edge sets during this process, so that $J_0 = E$ and $J_s = J$. We then use the method of conditional expectations to select a series of edge-sets $E_1, \dots, E_s$ to mimic $J_1, \dots, J_s$. 

Ideally, we would choose $E_i$ such that $\bE[ a(J') \mid J_i = E_i ] \geq \bE[ a(J') ]$.   In this conditional expectation, each edge $e \in E_i$ goes into $J$  with probability $2^i p$.  Unfortunately, the conditional expectation $\bE[ a(J') \mid J_i]$ is a complex, non-linear function of $J_i$.  We instead construct pessimistic estimators $S_0, \dots, S_s$ with $\bE[ a(L') \mid J_i = E_i] \approx S_i$, but which are more amenable to computations.

Now let us assume that we are given a good coloring of $H$ and edge-weight function $a$. For the formal construction, we define the potential functions
$$
S_i = (\tfrac{1}{2 \alpha})^{s-i} a(E_i) - b_i \sum_{v \in V} \Bigl( \tbinom{ \Delta 2^{-i} }{w} + \tbinom{\deg_{E_i}(v)}{w} \Bigr) a(N(v) \cap E_i)
$$
for each $i = 0, \dots, s$, where the parameters are defined as follows:
\begin{align*}
w &= \lceil 2 \log_2(r \log_2 \Delta) \rceil\\
s &= \big \lceil \log_2 \tfrac{\Delta}{ w^4 \log^{10}(r \Delta) } \big \rceil \\
\alpha &= 2^{1/s} \\
x &= \Delta 2^{-s}  \\
\beta &= 16 r (e x/w)^w \\
b_i &= 2^{-(s-i) (w+1)} \alpha^{s-i} / \beta
\end{align*}

In the definition of $\beta$, note that $ e = 2.71828...$ is the base of the natural logarithm.  For brevity, we also use the shorthand $d_i(v) = \deg_{E_i}(v)$ for vertex $v$.

The quantity $S_i$ is supposed to represent the expectation of $a(J')$, conditional on $J_i = E_i$. Note that for an edge $e$ remaining in $J_i$, the probability that $e$ survives to $J$ would be $p = 2^{-(s+i)} = 2^{i} x / \Delta$. With this in mind, let us provide some intuition for the different terms in $S_i$.  

The first term in $S_i$ represents the expected weight of the edges remaining in $J$, which is $a(E_i) p = a(E_i) 2^{-(s+i)}$. We include an additional error term $(\frac{1}{\alpha})^{s-i}$ here, because some edges will need to be discarded when we obtain our defective edge colorings.  

 The second term represents the total weight of the edges \emph{discarded} from $J'$ due to vertices with excessive degree. For a given vertex $v$, this expression is (up to scaling factors) given by:
$$
 \alpha^{s-i}  \times 2^{-(s-i) w} \Bigl( \tbinom{ \Delta 2^{-i} }{w} + \tbinom{d_i(v)}{w} \Bigr) \times 2^{-(s-i)} a(N(v) \cap E_i)
$$

\begin{itemize}
\item  The term $\alpha^{s-i}$ is a  fudge factor for some small multiplicative errors in our approximations.
\item The term $2^{-(s-i)} a(N(v) \cap E_i)$ is the expected value of $a(N(v) \cap J)$
\item The term $2^{-(s-i) w} \tbinom{d_i(v)}{w}$ is  the expected value of $\binom{\deg_J(v)}{w}$, and thus (up to rescaling) an approximation to the probability of discarding vertex $v$  from $J'$ due to having  $\deg_J(v) \gg x$.
\item The term $\binom{\Delta 2^{-i}}{w}$ is used to control the case where $d_i(v)$ is much smaller than $\Delta 2^{-i}$. In this situation, the term $\tbinom{d_i(v)}{w}$ will become negligible compared to $\tbinom{\Delta 2^{-i}}{w}$ and so we can essentially ignore $v$. 
\end{itemize}

We note that the hypergraph MIS algorithm of \cite{harris-mis} used  a similar derandomization algorithm with  a similar potential function as a pessimistic estimator. In particular, they developed the technical tool of the additive term to ``smooth'' errors in the multiplicative terms of concentration inequalities, which we adopt here. 

Our plan is to select a chain of edge subsets $E = E_0 \supseteq E_1 \supseteq \dots \supseteq E_s$, such that $S_0 \leq S_1 \leq \dots \leq S_s$. The key technical result for the algorithm will be the following: 
\begin{lemma}
\label{lem-splitg2}
 If $d_i(v) \leq \Delta 2^{4-i}$ for all vertices $v$, then there is a deterministic  $\tilde O(r + \log \Delta)$-round algorithm to find edge set $E_{i+1} \subseteq E_i$ such that $S_{i+1} \geq S_i$.
\end{lemma}

This lemma is quite involved; we show it next in Section~\ref{splitg2sec}. We will first show some straightforward properties of the potential function, and also show how Theorem~\ref{split2} follows from the lemma. At several places, we will use the elementary inequalities
\begin{eqnarray}
 p^w \tbinom{T}{w} \geq \tbinom{p T}{w} \qquad \qquad \forall p \in [0,1], T \in \mathbb Z_{+}  \label{yg5} \\
  (A/B)^B \leq  \tbinom{A}{B} \leq (e A / B)^B \qquad \text{for integers $A \geq B \geq 1$}   \label{yg52}
  \end{eqnarray}

\begin{proposition}
\label{s0prop}
We have $S_0 \geq \Omega( a(E) x / \Delta )$.
\end{proposition}
\begin{proof}
Note that $E_0 = E$. As $\alpha^s = 2$ and $2^{-s} = \frac{x}{\Delta}$,  at $i = 0$ we have
$$
(\tfrac{1}{2 \alpha})^{s-i} a(E) \geq 2^{-s-1} a(E) = \frac{a(E) x}{2 \Delta}
$$

Next, consider some vertex $v$, and we want to estimate the contribution of the term $b_i \bigl( \tbinom{ \Delta 2^{-i} }{w} + \tbinom{d_0(v)  }{w} \bigr) a(N(v) \cap E)$. We have $b_0 = 2^{-s(w+1)} \alpha^s / \beta = \tfrac{2}{\beta} (\tfrac{x}{\Delta})^{w+1}$ and $\tbinom{d_0(v)}{w} \leq \tbinom{\Delta}{w}$, so
$$
b_i  \bigl( \tbinom{ \Delta 2^{-i} }{w} + \tbinom{d_0(v) }{w} \bigr) \leq \tfrac{2}{\beta} (\tfrac{x}{\Delta})^{w+1} (\tbinom{ \Delta}{w} + \tbinom{\Delta}{w}) = \tfrac{4}{\beta} (\tfrac{x}{\Delta})^{w+1} \tbinom{ \Delta}{w}
$$

By Eq.~(\ref{yg52}), this is at most $\tfrac{4}{\beta} (\tfrac{x}{\Delta})^{w+1} (\tfrac{e \Delta}{w} )^w = \frac{4 x}{\beta \Delta} (\frac{e x}{w})^w$, which is equal to $\frac{x}{4 r \Delta}$ by definition of $\beta$. Using this formula, and the fact that $H$ has rank $r$, we then get:
\begin{align*}
&b_i \sum_{v \in V} \bigl( \tbinom{ \Delta 2^{-i} }{w} + \tbinom{d_i(v)  }{w} \bigr) a(N(v) \cap E) \leq  \frac{x}{4 r \Delta}  \sum_{v \in V} a(N(v) \cap E) \leq  \frac{x}{4 \Delta}  a(E)
\end{align*}

Thus  $S_0 \geq \frac{a(E) x}{\Delta} ( \frac{1}{2} - \frac{1}{4} ) \geq \Omega( a(E) x / \Delta ) $.
\end{proof}

\begin{proposition}
\label{vertex-discard-prop}
If we discard from edge set $E_i$ all vertices $v$ with $d_i(v) \geq \Delta 2^{4-i}$, along with their incident edges, then the value $S_i$ does not decrease. 
\end{proposition}
\begin{proof}
Let $U = \{ v \in V \mid d_i(v) \geq \Delta 2^{4-i} \}$, and we define $E'_i$ to be the edge-set $E_i$ after discarding the vertices in $U$. We also define $S'_i$ to be the resulting potential function, i.e. with $E'_i$ instead of $E_i$. Finally, write $d'_i(v) = \deg_{E'_i}(v) \leq \Delta 2^{4-i}$.

There are three main terms in the difference $S'_i - S_i$:
\begin{align*} 
  S_i' - S_i &= ( \tfrac{1}{2 \alpha})^{s-i} (a(E'_i) - a(E_i)) \\
  & + b_i  \sum_{v \in V-U} \bigl( \tbinom{ \Delta 2^{-i} }{w} + \tbinom{d_i(v)  }{w} \bigr) a(N(v) \cap E_i) - \bigl( \tbinom{ \Delta 2^{-i} }{w} + \tbinom{d'_i(v)  }{w} \bigr) a(N(v) \cap E'_i) \\
  & + b_i \sum_{v \in U} \bigl( \tbinom{ \Delta 2^{-i} }{w} + \tbinom{d_i(v)  }{w} \bigr) a(N(v) \cap E_i) - \bigl( \tbinom{ \Delta 2^{-i} }{w} + \tbinom{d'_i(v) }{w} \bigr) a(N(v) \cap E'_i) 
  \end{align*}

Let us estimate these in turn.  First, we have
  $$
  a(E'_i) - a(E_i) = -a(E_i - E'_i) = -a( \bigcup_{v \in U} N(v) \cap E_i ) \geq -\sum_{v \in U} a(N(v) \cap E_i)
  $$
  
  Next, for $v \in V - U$, we have $d_i(v) \geq d'_i(v)$ and $a(N(v) \cap E_i) \geq a(N(v) \cap E'_i)$, so 
  $$
  \bigl( \tbinom{ \Delta 2^{-i} }{w} + \tbinom{d_i(v)  }{w} \bigr) a(N(v) \cap E_i) - \bigl( \tbinom{ \Delta 2^{-i} }{w} + \tbinom{d'_i(v)  }{w} \bigr) a(N(v) \cap E'_i) \geq 0
  $$
  
  Finally, for $v \in U$, we have $N(v) \cap E'_i = \emptyset$ and $d_i(v) \geq \Delta 2^{4-i}$, so that
  \begin{align*}
&  \bigl( \tbinom{ \Delta 2^{-i} }{w} + \tbinom{d_i(v)  }{w}) a(N(v) \cap E_i) - \bigl( \tbinom{ \Delta 2^{-i} }{w} + \tbinom{d'_i(v)   }{w} \bigr) a(N(v) \cap E'_i) \geq \tbinom{\Delta 2^{4-i} }{w} a(N(v) \cap E_i)
 \end{align*}
 
 Putting these three terms together, we have shown that
 \begin{align*}
   S'_i - S_i &\geq ( \tfrac{1}{2 \alpha})^{s-i} \bigl( - \sum_{v \in U} a(N(v) \cap E_i) \bigr) + b_i  \sum_{v \in V-U} 0 + b_i \sum_{v \in U} \tbinom{\Delta 2^{4-i} }{w} a(N(v) \cap E_i) \\
  &= \sum_{v \in U}  \Bigl( -(\tfrac{1}{2 \alpha})^{s-i} + b_i  \tbinom{\Delta 2^{4-i} }{w} \Bigr) a(N(v) \cap E_i)
  \end{align*}
  
  In order to show that the sum is non-negative, we will show that
  \begin{equation}
    \label{yg6}
  b_i \tbinom{\Delta^{4-i} }{w}  \geq (\tfrac{1}{2 \alpha})^{s-i}
  \end{equation}
  
  Substituting the value of $b_i$, we calculate
  \begin{align*}
    \frac{b_i \tbinom{\Delta^{4-i} }{w}}{ (\tfrac{1}{2 \alpha})^{s-i} } &= \frac{ 2^{-(s-i) (w+1)}  \alpha^{s-i} \tbinom{\Delta^{4-i} }{w}}{ \beta (\tfrac{1}{2 \alpha})^{s-i}} = \frac{2^{-(s-i) w}  \alpha^{2(s-i)} \tbinom{\Delta^{4-i} }{w}}{ \beta }
  \end{align*}

  Using Eq.~(\ref{yg5}) and noting that $\alpha \geq 1$, we thus have
  \begin{align*}
    \frac{b_i \tbinom{\Delta^{4-i} }{w}}{ (\tfrac{1}{2 \alpha})^{s-i} } &\geq \frac{ 2^{-(s-i) w} \tbinom{\Delta 2^{4-i}}{w} }{\beta}  \geq \frac{ \tbinom{ \Delta 2^{4-i} \times 2^{-(s-i)}}{w}}{\beta} = \frac{ \tbinom{2^{4-s} \Delta}{w} }{\beta} = \frac{ \tbinom{16 x}{w} }{\beta}
  \end{align*}
      
  By Eq.~(\ref{yg52}),  we have $\tbinom{16 x}{w} \geq (16 x/w)^w$. Since $\beta = 16 r (e x/w)^w$, we have $\tbinom{16 x}{w} /  \beta \geq (16/e)^w / (16 r)$. From the definition of $w$ and our assumption that $r \geq 2$, this is greater than $1$.
   \end{proof}

\begin{proposition}
\label{hhreduce-prop}
There is a $\tilde O( r \log \Delta + \log^2 \Delta )$-round deterministic algorithm to find a subset $E' \subseteq E$ with maximum degree $\Delta' \leq \polylog(\Delta, r)$ and with $a(E') / \Delta' \geq \Omega( a(E)/\Delta)$.
\end{proposition}
\begin{proof}
We set $E_0 = E$ and proceed through $s$ stages; at the $i^{\text{th}}$ stage, we discard vertices with degree $d_i(v)$ exceeding $\Delta 2^{4-i}$ and we then apply Lemma~\ref{lem-splitg2} to $E_{i}$ to generate subset $E_{i+1} \subseteq E_{i}$ with $S_{i+1} \geq S_{i}$. We return the final set $E' = E_s$.   Each stage takes $\tilde O(r + \log \Delta)$ rounds and there are $s = O(\log \Delta)$ stages altogether, giving the stated complexity.

We now check the required bound on $a(E')$. Proposition~\ref{s0prop} shows that $S_0 \geq \Omega( a(E) x / \Delta)$. From Lemma~\ref{lem-splitg2} and Proposition~\ref{vertex-discard-prop}, we have $S_s \geq S_{s-1} \geq \dots \geq S_0$. Finally, we have
$$
S_s = a(E_s) - b_s \sum_{v \in V} \bigl( \tbinom{ \Delta 2^{-s} }{w} + \tbinom{d_s(v)  }{w} \bigr) a(N(v) \cap E_s) \leq a(E_s).
$$   

Putting these inequalities together, we therefore have $a(E_s) \geq \Omega( a(E) x / \Delta)$.

Our step of discarding high-degree vertices ensures that $E_s$ has maximum degree $\Delta' = \Delta 2^{4-s} = 16 x = \polylog(\Delta, r)$. Furthermore, we have
\[
a(E_s) / \Delta' \geq \frac{ \Omega(a(E) x/ \Delta) }{ 16 x } \geq \Omega( a(E)/\Delta) \qedhere
\]
\end{proof}

\begin{proof}[Proof of Theorem~\ref{split2}]
Given the initial hypergraph $H = (V,E)$ of maximum degree $\Delta$, we apply Proposition~\ref{hhreduce-prop} to obtain hypergraph $H' = (V,E')$ of maximum degree $\Delta' = \polylog(r, \Delta)$ and such that $a( E') / \Delta' \geq \Omega( a(E) / \Delta )$. Next, apply Proposition~\ref{hhreduce-prop} to hypergraph $H'$, obtaining a hypergraph $H'' = (V, E'')$ of maximum degree $\Delta'' = \polylog(r, \Delta') = \poly( \log r, \log \log \Delta)$ and such that $a(E'')/\Delta'' \geq \Omega( a(E')/\Delta') \geq \Omega( a(E)/\Delta)$.

Finally, apply Proposition~\ref{match-prop2} to hypergraph $H''$, obtaining a matching $M \subseteq E''$ such that $a(M) \geq \Omega( \tfrac{a(E'')}{r \Delta''}) \geq \Omega( \tfrac{a(E)}{r \Delta} )$.

The first application of Proposition~\ref{hhreduce-prop} takes $\tilde O( \log^2 \Delta + r \log \Delta)$ rounds. The second application takes $\tilde O( \log^2 \Delta' + r \log \Delta')$ rounds. The final application of Proposition~\ref{match-prop2} takes $\tilde O(r \Delta'')$ rounds. Noting that $\Delta'' = \poly( \log r, \log \log \Delta)$ and $\Delta' = \poly( \log r, \log \Delta )$, the overall complexity is $\tilde O( \log^2 \Delta + r \log \Delta)$.
\end{proof}

\section{Proof of Lemma~\ref{lem-splitg2}}
\label{splitg2sec}
For a given index $i < s$, our goal now is to find a subset $E_{i+1} \subseteq E_i$ with $S_{i+1} \geq S_i$.  As a starting point, consider the random process wherein each edge $e \in E_i$ goes into $E_{i+1}$ independently with probability $1/2$; it can be checked that $\bE[ S_{i+1} ] \geq S_i$.  Unfortunately, the potential function $S_{i+1}$ is not directly amenable to derandomization by Lemma~\ref{basicderand3}. We will develop an approximating statistic $\tilde S$, which is close to $S_{i+1}$, yet satisfies  properties (A1) and (A2). 

For this section, we will assume that $\Delta$ is larger than any needed constants; note that if $\Delta = O(1)$, then we can use \Cref{basicderand2} to find $E_{i+1}$ in $O(\Delta r) = O(r)$ rounds.

Lemma~\ref{basicderand3} requires an appropriate vertex coloring of the underlying graph $G$; in this case, this will correspond to a defective edge-coloring of $H$. The first step in the derandomization is to use Lemma~\ref{lem-split-crude} to obtain a partial balanced edge-coloring $\chi$, which retains most of the edges by weight. We summarize this in the following result:

\begin{proposition}
  \label{lem-splitg21}
  In $\tilde O(r \polyloglog \Delta)$ rounds, we can generate an edge-set $F \subseteq E_i$ with $a(F) \geq a(E_i)/\alpha$ along with a $t$-defective $k$-edge-coloring $\chi$ of $F$, where we define the parameters
$$
k = \lceil 2048 w^2 \log \Delta \rceil, \qquad t = 2^{6 - i} \Delta / k.
$$
\end{proposition}
\begin{proof}
Here, $E_i$ has maximum degree $\Delta' = \Delta 2^{4-i}$.  We will apply Lemma~\ref{lem-split-crude} to $E_i$ with parameters $k$ and $\delta =  1 - 1/\alpha$. (We assume here that we have been provided a good coloring of $H$.)    We must check that the hypotheses of Lemma~\ref{lem-split-crude} are satisfied, specifically, we need  $\Delta' \geq C k \log( \tfrac{r \log k}{\delta})$. We have $\Delta'  \geq \Delta 2^{4-s} = 16 x$. Also, $\frac{1}{\delta} =  \frac{\alpha}{\alpha - 1} \leq \frac{1}{2^{1/s} - 1} = O(s) = O(\log \Delta)$. Thus, it suffices to show that $x \geq C' k \log(r \Delta \log k)$ for some constant $C'$. Since $k = O( w^2 \log \Delta)$ and $x = \Theta(w^4 \log^{10}(r \Delta))$, this indeed holds for $\Delta$ sufficiently large.

  This gives a coloring $\chi$ of an edge-set $F \subseteq E_i$ with $a(F) \geq (1-\delta) a(E_i)$ which has $k$ colors and has defectiveness $4 \Delta'/k = t$. Lemma~\ref{lem-split-crude} runs in $\tilde O(r \log \tfrac{1}{\delta} \log^3 k) = \tilde O( r \polyloglog \Delta)$ rounds.
\end{proof}

Given $F$ and $\chi$, we come to the heart of the construction. Let us first make a number of definitions.  For a vertex $v \in V$ and edge $e \in F$ define $\mathcal U_{v,e} \subseteq \binom{ N(v) \cap F}{w}$ to be the set of $w$-element set $W = \{ f_1, \dots, f_w \} \subseteq N(v) \cap F$ with the property that the values $\chi(e), \chi(f_1), \dots, \chi(f_w)$ are all distinct. We also define $R_{v,e} = \bigl| \mathcal U_{v,e} \cap 2^{E_{i+1}} \bigr|$.

 We now  define $\tilde S$ as a function of $E_{i+1}$ as follows:
  $$
  \tilde S = ( \tfrac{1}{2 \alpha})^{s-(i+1)} a(E_{i+1}) - \alpha b_{i+1}  \sum_{v \in V} \sum_{e \in N(v) \cap E_{i+1}} \Bigl( R_{v,e} + \tbinom{\Delta 2^{-(i+1)}}{w} \Bigr) a(e)
$$

We will apply Lemma~\ref{basicderand3} with respect to statistic $\tilde S$. This gives an edge-set $E_{i+1} \subseteq F$ with $\tilde S \geq \bE[ \tilde S ]$; here the expectation is taken over the random process wherein edges of $F$ go into $E_{i+1}$ independently with probability $1/2$. In order to show that $E_{i+1}$ has the desired properties, we will show the following chain of inequalities:
\begin{equation}
\label{ineq-chain}
S_{i+1} \geq \tilde S \geq \bE[\tilde S] \geq S_i
\end{equation}

We break this down into a number of smaller claims. 
\begin{proposition}
\label{mm3}
We have $S_{i+1} \geq \tilde S$.
\end{proposition}
\begin{proof}
We compute the difference:
{\allowdisplaybreaks
\begin{align*}
S_{i+1} - \tilde S &= -b_{i+1} \sum_{v \in V} \bigl( \tbinom{d_{i+1}(v)}{w} + \tbinom{\Delta 2^{-(i+1)}}{w} \bigr) a(N(v) \cap E_{i+1}) \\
& \qquad \qquad + \alpha b_{i+1}  \sum_{v \in V} \sum_{e \in N(v) \cap E_{i+1}} \bigl( R_{v,e} + \tbinom{\Delta 2^{-(i+1)}}{w} \bigr) a(e) \\
&= b_{i+1} \sum_{v \in V} \sum_{e \in N(v) \cap E_{i+1}} a(e) \Bigl( \alpha R_{v,e} + (\alpha-1) \tbinom{\Delta 2^{-(i+1)}}{w} - \tbinom{d_{i+1}(v)}{w} \Bigr) 
\end{align*}
}

To show this is non-negative, we claim that for any vertex  $v$ and edge $e \in N(v) \cap E_{i+1}$ we have 
\begin{equation}
\label{hnn1}
\alpha R_{v,e} + (\alpha - 1) \tbinom{\Delta 2^{-(i+1)}}{w} \geq \tbinom{d_{i+1}(v)}{w}
\end{equation}
Let $y = d_{i+1}(v)$. There are two cases to show Eq.~(\ref{hnn1}).

\noindent \textbf{Case I: $y \leq \Delta 2^{-(i+2)}$.} Then it suffices to show that $ (\alpha - 1) \tbinom{\Delta 2^{-(i+1)}}{w} \geq \tbinom{\Delta 2^{-(i+2)}}{w}$.

Since $\frac{ \tbinom{\Delta 2^{-(i+1)}}{w} }{ \tbinom{\Delta 2^{-(i+2)}}{w} } \geq (\frac{ \Delta 2^{-(i+1)} }{ \Delta 2^{-(i+2)} })^w = 2^w$ and $(\alpha - 1) = 2^{1/s} - 1 \geq \frac{1}{2 s}$, it suffices to show that $2^w \geq 2 s$. This holds for $\Delta$ sufficiently large because $s \leq 1+ \log_2 \Delta$ and $w \geq 2 \log_2 \log_2 \Delta$.

\noindent \textbf{Case II:  $y > \Delta 2^{-(i+2)}$.} Then, in order to show Eq.~(\ref{hnn1}), it suffices to show that $\alpha R_{v,e} \geq \tbinom{y}{w}$.

Note here that we have
\begin{equation}
\label{hnn23}
\frac{w t}{y} \leq \frac{4 w (\Delta 2^{6-i})/k}{\Delta 2^{-(i+2)}} = \frac{256 w}{k} \leq \frac{256 w}{2048 w^2 \log \Delta} = \frac{1}{8 w \log \Delta}
\end{equation}

In particular, $y \geq w t$. We now claim that we have the bound:
\begin{equation}
\label{uve-bnd}
R_{v,e} \geq (y - t w)^w/w!
\end{equation}

To show Eq.~(\ref{uve-bnd}), note that we can construct a set $W \in \mathcal U_{v,e} \cap 2^{E_{i+1}}$ as follows. First, select some edge $f_1$ in $N(v) \cap E_{i+1}$ with a different color than $e$. Since there are at most $t$ edges with the same color as $e$, there are at least $y - t$ such choices. Next, select edge $f_2 \in N(v) \cap E_{i+1}$ with a different color than $e$ or $f_1$. Again, there are at least $y- 2 t$ such choices. Continue this process to choose edges $f_3, \dots, f_w$; we will have at least $y - t w \geq 0$ choices for each edge $f_j$ in this process. Now observe that any set of edges $W = \{ f_1, \dots, f_w \}$ is counted $w!$ times in this enumeration process. 

As $\tbinom{y}{w} \leq y^w / w!$, the bounds from Eqs.~(\ref{hnn23}, \ref{uve-bnd}) show that:
$$
\frac{ R_{v,e} }{ \tbinom{y}{w} } \geq \frac{ (y - t w)^w / w! }{ y^w / w!} = (1 - \frac{w t}{y})^w \geq (1 - \frac{1}{8 w \log \Delta})^w
$$

We have $1 - \frac{1}{8 w \log \Delta} \geq e^{-\frac{1}{4 w \log \Delta}}$ and $s \leq 1 + \log_2 \Delta$, so for $\Delta$ sufficiently large we have
\[
\frac{ \alpha R_{v,e} }{ \tbinom{y}{w} } \geq \alpha (e^{-\frac{1}{4 w \log \Delta}})^w = 2^{1/s} \times e^{-\frac{1}{4 \log \Delta}} = e^{\frac{\log 2}{s} - \frac{1}{4 \log \Delta}} \geq 1. \qedhere
\]
\end{proof}

\begin{proposition}
\label{mm1}
For the random process wherein each edge $e \in F$ goes into $E_{i+1}$ independently with probability $1/2$, we have $\bE[\tilde S] \geq S_i$.
\end{proposition}
\begin{proof}
We compute $E[ \tilde S]$ as:
\begin{align*}
( \tfrac{1}{2 \alpha} )^{s-(i+1)} \bE[a(E_{i+1})] - \alpha b_{i+1} \sum_{v \in V} \sum_{e \in N(v) \cap F} \negthickspace\negthickspace \Pr(e \in E_{i+1})  a(e) \Bigl( \bE [ R_{v,e}  \mid e \in E_{i+1} ] + \tbinom{\Delta 2^{-(i+1)}}{w} \Bigr)
\end{align*}

Consider a vertex $v$ and edge $e \in N(v) \cap F$. The edge $e$ goes into $E_{i+1}$ with probability $\tfrac{1}{2}$. Since $e \notin W$ for any $W \in \mathcal U_{v,e}$, the conditional probability of $W$ surviving to $2^{E_{i+1}}$ is exactly $2^{-w}$, and so
\begin{equation}
\label{bgg4}
\bE \bigl[ R_{v,e} \mid e \in E_{i+1}  \bigr] = |\mathcal U_{v,e}| 2^{-w}
\end{equation}

Here $\mathcal U_{v,e} \subseteq \binom{ N(v) \cap E_i}{w}$, so $|\mathcal U_{v,e}| \leq \tbinom{ d_i(v)}{w}$. Also, by Eq.~(\ref{yg5}), we have $\tbinom{\Delta 2^{-(i+1)}}{w} \leq 2^{-w} \tbinom{\Delta 2^{-i}}{w}$. With this inequality and Eq.~(\ref{bgg4}), we have
$$
\bE \bigl[ R_{v,e} \mid e \in E_{i+1}  \bigr] + \tbinom{\Delta 2^{-(i+1)}}{w} \leq 2^{-w} \Bigl( \tbinom{d_i(v)}{w} + \tbinom{ \Delta 2^{-i}}{w} \Bigr)
$$

In addition, we have $\bE[a(E_{i+1})] = \frac{a(F)}{2} \geq \tfrac{a(E_i) }{2 \alpha}$, so $(\tfrac{1}{2 \alpha})^{s-(i+1)} \bE[a(E_{i+1})] \geq (\tfrac{1}{2 \alpha})^{s-i} a(E_i)$.

Putting these contributions together, we see:
\begin{align*}
\bE[\tilde S ] &\geq  ( \tfrac{1}{2 \alpha} )^{s-i} a(E_i) - \alpha b_{i+1} \sum_{v \in V} \sum_{e \in N(v) \cap F}  2^{-w-1}  \Bigl(  \tbinom{d_i(v)}{w} + \tbinom{ \Delta 2^{-i}}{w} \Bigr) a(e) \\
&=   ( \tfrac{1}{2 \alpha} )^{s-i} a(E_i) - \alpha  2^{-w-1} b_{i+1} \sum_{v \in V}  \Bigl( \tbinom{d_i(v)}{w} + \tbinom{ \Delta 2^{-i}}{w} \Bigr) a( N(v) \cap F)
\end{align*}

By direct calculation we see that $\alpha 2^{-w-1} b_{i+1} = b_i$. Also $a(N(v) \cap F) \leq a(N(v) \cap E_i)$, so we have the lower bound: 
\[
\bE[\tilde S] \geq   ( \tfrac{1}{2 \alpha} )^{s-i} a(E_i) - b_i \sum_{v \in V}  \Bigl( \tbinom{d_i(v)}{w} + \tbinom{ \Delta 2^{-i}}{w} \Bigr) a( N(v) \cap E_i) = S_i \qedhere
\]
\end{proof}

\begin{lemma}
\label{mm2}
The set $E_{i+1}$ can be generated in $\tilde O( \log \Delta \polylog r)$ rounds such that $\tilde S \geq \bE[\tilde S]$.
\end{lemma}
\begin{proof}
Let us define the graph $G = (\text{Inc}(H'))^2$ where $H'$ is the hypergraph $(V,F)$.  Note that $\chi$ can be viewed as a vertex coloring of $G$. We will apply Lemma~\ref{basicderand3} to the potential function $\tilde S$ with respect to $G$ and $\chi$. 
The runtime of Lemma~\ref{basicderand3} will be $O(k) = O(w^2 \log \Delta) = \tilde O( \log \Delta \polylog r)$.  We need to show that the potential function $\tilde S$ satisfies criteria (A1) and (A2). 

Recall that the random process  is that each edge of $F$ goes into $E_{i+1}$ independently with probability $1/2$. More concretely, let us say that each edge $e \in E_{i+1}$ chooses a $1$-bit random quantity $X_e$, and goes into $E_{i+1}$ if $X_e = 1$. We have
$$
R_{v,e} = \sum_{W \in \mathcal U_{v,e}} \prod_{f \in W} X_f
$$
so we can write $\tilde S$ as a polynomial in the $X_e$ variables as:
$$
\tilde S = (\tfrac{ 1}{2 \alpha})^{s-(i+1)} \sum_{e \in F} a(e) X_e - \alpha b_{i+1} \sum_{v \in V} \sum_{e \in N(v) \cap F} a(e) X_e \Bigl( \tbinom{ \Delta 2^{-(i+1)}}{w} + \sum_{W \in \mathcal U_{v,e}} \prod_{f \in W} X_f \Bigr)
$$

We want to compute the derivative $D_{g} \tilde S$ for some edge $g \in F$. By the linearity of the differentiation operator, we can calculate:
\begin{align*}
D_{g} \tilde S &= (\tfrac{ 1}{2 \alpha})^{s-(i+1)}  a(g)  -\alpha b_{i+1} \sum_{v \in V: g \in N(v)} a(g) \Bigl( \tbinom{ \Delta 2^{-(i+1)}}{w} + \sum_{W \in \mathcal U_{v,g}} \prod_{f \in W} X_f \Bigr) \\
& \qquad \qquad \qquad - \alpha b_{i+1} \sum_{v \in V} \sum_{\substack{e \in N(v) \cap F \\ e \neq g}} a(e) \sum_{\substack{W \in \mathcal U_{v,e} \\ g \in W}} \prod_{f \in W - \{g \}} X_f  
\end{align*}

This quantity only depends on the values $X_f$ for edges $f$ such that $f \in W \in \mathcal U_{v,g}$ or such that there is an edge $e$ with $\{f, g \} \subseteq W \in \mathcal U_{v,e}$. 

To show (A1), note that all such edges $f$ have a common vertex $v$ with $g$, and so the nodes are adjacent in $G^2$. Hence,  $g$ can locally compute the value $D_{g} \tilde S$.

To show (A2), note that by definition of $\mathcal U_{v,g}$ and $\mathcal U_{v,e}$, we must have $\chi(f) \neq \chi(g)$. Hence, the value of $D_g \tilde S$ is not affected by $X_f$ for edges $f$ with $\chi(f) = \chi(g)$.
\end{proof}

At this point, we have shown all the inequalities needed in Eq.~(\ref{ineq-chain}). Generating the coloring $\chi$ takes $\tilde O(r \polylog \Delta)$ rounds and applying Lemma~\ref{basicderand3} takes $\tilde O(\log \Delta \polylog r)$ rounds. This concludes the proof of Lemma~\ref{lem-splitg2}.

\section{Finding fractional hypergraph matchings}
\label{find-match}
The hypergraph matching algorithms of Theorems~\ref{main-thm1} and \ref{main-thm2} start by finding a high-weight fractional matching, which we describe in this section.  It is critical here to keep track of how close the fractional matching is to being integral. We use the following definition: 
\begin{definition}[$q$-proper fractional matching]
A fractional matching $h: E \rightarrow [0,1]$ is \emph{$q$-proper} if all the entries of $h$ are rational numbers with denominator $q$.
\end{definition}

Thus, an integral matching is a $1$-proper fractional matching.  There is a simple correspondence between $q$-proper fractional matchings and degree-$q$ hypergraphs: for a hypergraph $H$ and a $q$-proper fractional matching $h$, we define the  \emph{replicate hypergraph} $H^{[h]}$ by taking $q h(e)$ copies of each edge $e$. Note that the edge-set of $H^{[h]}$ is a multi-set. This hypergraph $H^{[h]}$ has maximum degree $q$ and has $a(H^{[h]}) = q a(h)$. This correspondence goes the other way as well: for a hypergraph $H$ of maximum degree $\Delta$, the fractional matching which assigns $h(e) = 1/\Delta$ for every edge is a $\Delta$-proper fractional matching.

We  obtain the fractional matching by using an algorithm of \cite{kmw} for packing or covering LP systems as well as some techniques for quantizing edge weights inspired by \cite{lotker}. These are relatively routine details so we defer the proof to Appendix~\ref{kmw-appendix}. We summarize this as follows:
\begin{lemma}
  \label{find-fmatch1}
Let $H$ be a hypergraph with an edge-weighting $a$. Recall that $a^*(H)$ is the maximum weight fractional matching for $H$.
  \begin{enumerate}   
  \item There is a deterministic $O(\log^2(\Delta r))$-round algorithm to generate a fractional matching $h$ which is $O(\Delta)$-proper and which satisfies $a(h) \geq \Omega(a^*(H))$.
    \item There is a randomized $O(\log r \log(\Delta r))$-round algorithm to generate a fractional matching $h$ which is $O(\log r)$-proper and which satisfies $\bE[a(h)] \geq \Omega(a^*(H))$.
  \end{enumerate}
\end{lemma}

This immediately gives our main deterministic algorithm for hypergraph matching:
\begin{reptheorem}{main-thm1}
For a hypergraph $H$ with a good coloring and an edge-weighting $a$, there is a deterministic $\tilde O(r \log \Delta + \log^2 \Delta)$-round  algorithm to generate a matching $M$ with $a(M) \geq \Omega(a^*(H)/r)$.
\end{reptheorem}
\begin{proof}
  Use \Cref{find-fmatch1} to obtain a fractional matching $h$ which is $q$-proper with $a(h) \geq \Omega(a^*(H))$ and $q = O(\Delta)$ in $O(\log^2 (\Delta r))$ rounds. Next, apply \Cref{split2} to hypergraph $H^{[h]}$, which has maximum degree $q$, to get a matching $M$ with $a(M) \geq \Omega(\frac{a(H^{[h]})}{q r}) = \Omega(a^*(H) / r)$. This takes $\tilde O(\log^2 q + r \log q) = \tilde O(\log^2 \Delta + r \log \Delta)$ rounds.
\end{proof}

We next turn to the randomized algorithm. To emphasize it is truly local, we show that our algorithm achieves success probability of $1-\delta$ for an arbitrary parameter $\delta$, which may depend on $n$ or any other quantities.  The strategy is to first randomly sparsify the hypergraph so that $\Delta \approx \polylog \tfrac{1}{\delta}$, and then use our deterministic algorithm. Note that this strategy is very different from a straightforward simulation of direct rounding. 

\begin{reptheorem}{main-thm2}
Let $\delta \in (0,\tfrac{1}{2})$ be an arbitrary parameter.  For a hypergraph $H$ with a good coloring an an edge-weighting $a$,  there is a randomized $\tilde O(\log \Delta + r \log \log \tfrac{1}{\delta} + (\log \log \tfrac{1}{\delta})^2)$-round  algorithm to generate a matching $M$ such that $a(M) \geq \Omega(a^*(H)/r)$ with probability at least $1 - \delta$.
\end{reptheorem}
\begin{proof}
Execute $t$ independent parallel applications of the randomized part of Lemma~\ref{find-fmatch1}, for a parameter $t > 0$ to be determined. This runs in $O( \log r \log(\Delta r))$ rounds and produces fractional matchings $h_1, \dots, h_t$ which are each $q$-proper for $q = O( \log r)$ and with $\bE[a(h_i)] \geq \Omega(a^*(H))$.

  Now form the fractional matching $h = (h_1 + \dots + h_t) / t$,  and consider the hypergraph $H' = H^{[h]}$. Since $h$ is $qt$-proper, $H'$ has maximum degree $\Delta' = q t$, and each $h_i$ is a fractional matching of $H'$.

We have $\bE[a(h_i)] \geq \Omega(a^*(H))$, and $a(h_i) \leq a^*(H)$ with probability one by definition of $a^*(H)$. Markov's inequality applied to the non-negative random variable $a^*(H) - a(h_i)$ shows that $a(h_i) \geq \Omega(a^*(H))$ with probability $\Omega(1)$. Since $a(h_1), \dots, a(h_t)$ are independent random variables, there is a probability of at least $1 - 2^{-\Omega(t)}$ that $a(h_i) \geq \Omega(a^*(H))$ for at least one value of $i$. Since each $h_i$ is a fractional matching of $H'$, this means that by taking $t = \lceil c \log \tfrac{1}{\delta} \rceil$ for a sufficiently large constant $c$, we ensure that $a^*(H') \geq \Omega(a^*(H))$ holds with probability at least $1 - \delta/2$. 

Next,  randomly choose an edge-coloring of $H'$ with $4 \Delta' r / \delta$ colors, and discard all pairs of adjacent edges with the same color. We claim that, with probability $1 - \delta/2$, the resulting hypergraph $H''$ has $a^*(H'') \geq \Omega(a^*(H'))$. For, consider an optimal fractional matching $g$ of $H'$. Each edge is discarded in $H''$ with probability at most $\delta/4$, and so the expected weight of discarded edges from $g$ i is at most $\delta a^*(H)  / 4$; by Markov's inequality, it is at most $a^*(H)/2$ with probability at least $1-\delta/2$. In this case, fractional matching $g$ has weight $a^*(H')/2$ in $H''$.

Now suppose that both desired events have have occurred, and so $a^*(H'') \geq \Omega(a^*(H')) \geq \Omega(a^*(H))$.  In $O(\log^* \tfrac{t r}{\delta})$ rounds the coloring of $H''$ can be converted into a good coloring of $H''$. We finish by applying \Cref{main-thm1} to $H''$ get a matching $M$  with $a(M) \geq \Omega(a^*(H''))$.  This step takes $\tilde O(r \log \Delta' + \log^2 \Delta') = \tilde O(r \log \log \tfrac{1}{\delta} + (\log \log \frac{1}{\delta})^2)$ rounds.
\end{proof}

\section{Maximum-weight graph matching}
\label{graph-match}
Consider a graph $G$  with an edge-weighting  $a$. Our goal is to find a matching $M$ with $a(M) \geq (1-\epsilon) T$ for some desired parameter $\epsilon > 0$, where $T$ denotes the maximum weight matching. We refer to such $M$ as an \emph{$\epsilon$-near matching} of $G$.

  The overall plan is to iteratively improve the matching by augmenting it with short alternating paths. This basic idea has been used for parallel algorithms in \cite{hv} and adapted to the distributed  setting  in \cite{nieberg, GKMU18}. Formally, for a matching $M$ of $G$, we define an \emph{$\ell$-augmentation $P$} to be a path or cycle of length at most $2 \ell$ which alternately passes through matched and unmatched edges, and has the additional property that if it ends at an edge $(u,v) \in E - M$ or starts at an edge $(v,u) \in E - M$, then  vertex $v$ must be unmatched in $M$.

We can augment the matching with respect to $P$, obtaining a new matching $M' = M \oplus P$. We define the \emph{gain} of $P$ as $g(P) = a(M') - a(M) = a(P - M) - a(P \cap M)$. If we have a collection $\mathcal P$ of vertex-disjoint paths or cycles, then we can augment them all simultaneously, getting a new matching $M'$ with $a(M') = a(M) + g(\mathcal P)$ where we define $g(\mathcal P) = \sum_{P \in \mathcal P} g(P)$.  We quote the following result from \cite{pettie-sanders}:
  \begin{proposition}[\cite{pettie-sanders}]
    \label{nieberg-prop}
    Let $M$ be an arbitrary matching of $G$. For an integer $\ell \geq 1$, there is a collection $\mathcal P$ of vertex-disjoint $\ell$-augmentations with $g(\mathcal P) \geq \tfrac{1}{2} ( T (1 - 1/\ell) - a(M) ).$
  \end{proposition}

 The algorithm for graph matching, as we summarize next, is based on representing augmentations as hypergraph matchings.
\begin{reptheorem}{mt1}
Let $\epsilon \in (0,1)$ and let $G$ be a graph of maximum degree $\Delta$ with an edge-weighting $a$. 
  \begin{enumerate}
    \item There is a deterministic $\tilde O(\epsilon^{-4} \log^2 \Delta + \epsilon^{-1} \log^* n)$-round algorithm to find an $\epsilon$-near graph matching.

    \item   For any $\delta \in (0,\tfrac{1}{2})$, there is a randomized $\tilde O( \epsilon^{-3} \log \Delta + \epsilon^{-3} \log \log \tfrac{1}{\delta} + \epsilon^{-2} (\log \log \tfrac{1}{\delta})^2 )$-round algorithm to find an $\epsilon$-near graph matching  with  probability at least $1 - \delta$.
      \end{enumerate}
\end{reptheorem}
\begin{proof}  
  We first describe the deterministic algorithm. Let us set $\ell = \lceil 2/\epsilon \rceil$, and we define the \emph{path hypergraph} $\overline H$ to have vertex set $V$ and to have a hyperedge $\{v_1, \dots, v_s \}$ for every path or cycle  $(v_1, v_2, \dots, v_{s})$ of length $s \leq 2 \ell$ in $G$. Here $\overline H$ has rank $2 \ell$ and maximum degree  $\Delta^{2 \ell}$, and a matching of it corresponds to a collection of length-$2 \ell$ vertex-disjoint paths in $G$. 

  Our first step is to get a good coloring of $\overline H$. Next, we start with matching $M_0 = \emptyset$ and go through $t$ augmentation stages.  In stage $i < t$, we form hypergraph $H_i$ whose edges are $\ell$-augmentations of $M$, and where the weight of an edge of $H_i$ is the gain of the corresponding augmentation. Note that $H_i$ is a sub-hypergraph of $\overline H$. We apply \Cref{split2} to $H_i$, using the given coloring of $\overline H$, obtaining a matching $N_i$ of $H_i$. We then form $M_{i+1}$ by augmenting $M_i$ with respect to $N_i$.  At the end of this process, we output the final matching $M = M_{t+1}$.

  Let us first show that the resulting matching $M$ is $\epsilon$-near. Define $g_i$ to be the gain of set of paths $N_i$ with respect to matching $M_i$ and let us define $\alpha_i = T(1 - \tfrac{\epsilon}{2}) - a(M_i)$. By Proposition~\ref{nieberg-prop}, each $H_i$ has a fractional matching whose gain (with respect to $M_i$) is at least $\tfrac{1}{2} ( T (1 - 1/\ell) - a(M_i) ) \geq \alpha_i / 2$. Since $H_i$ has rank $2 \ell$, this means that $g_i \geq \Omega( \alpha_i / \ell ) = \Omega( \epsilon \alpha_i)$. So we have
  $$
  a(M_{i+1}) = a(M_i) + g_i \geq a(M_i) + \Omega(\epsilon \alpha_i)
  $$

Equivalently, we have $\alpha_{i+1} \leq \alpha_i ( 1 - \Omega(\epsilon))$.  Since $\alpha_0 \leq T$, this implies that $\alpha_{t+1} \leq \epsilon/2$ for some $t = \Theta( \tfrac{ \log(1/\epsilon) }{\epsilon} )$. This implies that $a(M_{t+1}) = T(1 - \epsilon/2) - \alpha_{t+1} \geq T (1-\epsilon)$, as desired.

 Let us next examine the complexity of this process. It requires $O( \tfrac{\log^*n}{\epsilon})$ rounds to get the coloring of $\overline H$. In each stage $i$, the hypergraph $H_i$ has maximum degree $\Delta' = \Delta^{2 \ell}$ and rank $2 \ell$. \Cref{main-thm1} requires $\tilde O(\ell \log \Delta' + \log^2 \Delta') = \tilde O(\tfrac{\log^2 \Delta}{\epsilon^2})$ rounds on $H_i$. Each communication step of $H_i$ can be simulated in $O(\ell) = O(1/\epsilon)$ rounds on $G$, so this overall takes $\tilde O(\tfrac{\log^2 \Delta}{\epsilon^3})$ time per stage. There are $t = \tilde O(1/\epsilon)$ stages altogether.

  The randomized version is completely analogous, except that we do not obtain the good coloring of $\overline H$ and we use \Cref{main-thm2} (with failure probability $\delta' = \delta/t)$ instead of \Cref{main-thm1}.
\end{proof}

Let us compare our algorithm with known bounds for approximating GMWM.

First, \cite{kmw} showed a lower bound of $\Omega ( \min (\sqrt{ \frac{\log n}{ \log \log n}},  \frac{\log \Delta}{\log \log \Delta} ) )$ rounds for any constant-factor approximation. There is also a randomized algorithm of \cite{bcgs}, which runs in $O(\frac{\log \Delta}{\epsilon^{3} \log \log \Delta})$ and gives an $\epsilon$-near matching with constant probability. Thus, the randomized round complexity of approximate maximum matching is precisely $\Theta( \frac{\log \Delta}{\log \log \Delta} )$ as a function of $\Delta$. Our randomized algorithm matches this up to $\log \log \Delta$ factors, and our deterministic algorithm matches this up to a factor of $\log \Delta$. 

Additionally, \cite{bks} showed an $\Omega(1/\epsilon)$ lower bound on the run-time for deterministic or randomized algorithms to get $\epsilon$-near matchings. 

Finally, we give a simple reduction to $3$-coloring a ring graph to show that the $\log^* n$ term is needed, even for Maximum Cardinality Matching with $\Delta = 2$. 
\begin{theorem}
  \label{gmwm-lb-thm}
  Let $\rho, n \geq 1$.
Any deterministic LOCAL algorithm for $\rho$-approximate graph Maximum Cardinality Matching on $n$-vertex graphs requires $\Omega(\frac{\log^* n}{\rho})$ rounds.
\end{theorem}
\begin{proof}
  Consider an $n$-vertex ring graph $G$, and suppose that algorithm $A$ runs in $t$ rounds on $G$ and guarantees a matching $M$ which is a $\rho$-approximation to maximum cardinality matching. We first claim that every contiguous sequence of $8 \rho t$ edges has at least one edge in $M$.

  For, suppose that a contiguous sequence $v_1, \dots, v_{\ell}$ lacks such an edge.  Form the ring graph $G'$ on vertices $v_1, \dots, v_{\ell}$ (the vertex $v_{\ell}$ becomes joined to the vertex $v_1$).  The vertices $v_t, \dots, v_{\ell -t}$ will not see any difference in their $t$-neighborhood compared to $G$, and so when we run $A$ on $G'$, the resulting matching $M'$ will not have any edges between $v_t, \dots, v_{\ell - t}$. Therefore, $|M'| \leq 2\lceil t/2 \rceil \leq t+1$. On the other hand, a maximum matching of $G'$ has size $\lfloor \ell/2 \rfloor \geq \ell/2 - 1$. Since $A$ guarantees a $\rho$-approximation, we must have have $t+1 \geq \frac{(\ell/2) - 1}{\rho}$, i.e. $\ell \leq 2 (1 + \rho + \rho t) < 8 \rho t$.

  Now, using $A$, we can form an $8 \rho t + 1$-ruling set $U$ for $G$, by selecting the lower-ID vertex of each edge of $M$. This allows us to generate a $3$-coloring of $G$ in $O(\rho t)$ rounds: we sort the vertices by their increasing distance from the closest element of $U$, and at each stage $i = 0, \dots, 8 \rho t+1$, the vertices at distance $i$ greedily choose a color. 

  On the other hand, Linial \cite{lin92} showed that $3$-coloring a ring graph requires $\Omega(\log^* n)$ rounds. Thus $t \geq \Omega( \frac{\log^* n}{\rho})$.
\end{proof}

\section{Hypergraph maximal matching and applications}
\label{maximal-match}
As is standard for matching algorithms, we build a maximal matching by repeatedly finding residual matchings. For a matching $M$ of a hypergraph $H$, the \emph{residual hypergraph} $\text{Res}_M(H)$ is constructed by removing all edges from $H$ which intersect with $M$. Note that $M$ is a maximal matching if and only if $\text{Res}_M(H)$ has no edges. We also define $\tau(H)$ to be the size of the largest matching in $H$.

The simplest hypergraph maximal matching (HMM) algorithm comes directly from the approximate maximum matching algorithm, as follows:
\begin{theorem}
\label{mt2a}
There is a deterministic HMM algorithm in $\tilde O( (r \log \tau(H)) (r \log \Delta + \log^2 \Delta) + \log^* n)$ rounds. There is a randomized HMM algorithm with success probability at least $1 - \delta$ in $\tilde O( (r \log \tau(H)) (\log \Delta + r \log \log \tfrac{1}{\delta} + (\log \log \tfrac{1}{\delta})^2) )$ rounds.
\end{theorem}
\begin{proof}
  We first describe the deterministic algorithm. To begin, we compute a good coloring of $H$. We next initialize $M_0 = \emptyset$ and go through $t$ stages, where in each stage $i$ we apply \Cref{main-thm1} to hypergraph $\text{Res}_{M_i}(H)$ with the constant edge-weighting function to obtain a matching $L_i$. We form $M_{i+1} = M_i \cup L_i$, which is a matching of $H$ by definition of the residual hypergraph.

Define $\tau_i = \tau( \text{Res}_H(M_i) )$.  Since we are using a constant edge-weighting function, \Cref{main-thm1} ensures that $|L_i| \geq \Omega( \tau_i / r )$. Any matching of $\text{Res}_H(M_{i+1})$ could be combined with $L_i$ to yield a matching of $\text{Res}_H(M_i)$, so $\tau_{i+1} \leq \tau_i - |L_i| \leq \tau_i (1 - \Omega(1/r))$. This implies that $\tau_{t+1} < 1$ for $t = \Theta( r \log \tau(H) )$ and so matching $M_{t+1}$ is maximal.

 It requires $O(\log^* n)$ rounds to get the coloring of $H$. Each stage of \Cref{main-thm1} runs in $\tilde O(r \log \Delta + \log^2 \Delta)$ rounds, and there are $t = O( r \log \tau(H) )$ stages altogether.

   The randomized algorithm is completely analogous, using the randomized version of \Cref{main-thm1} with appropriately chosen failure probability $\delta' = \delta/\poly(r, \log \tau(H), \log \Delta, \log \tfrac{1}{\delta})$.
    \end{proof}

We also describe a second, alternative algorithm based on the ``shattering'' technique of \cite{beps}:
\begin{theorem}
\label{split4r}
There is an $\tilde O(r \log^2 \Delta+ r^2 (\log \log n)^2 + r (\log \log n)^3)$-round randomized algorithm to get a maximal matching of a hypergraph $H$ w.h.p.
\end{theorem}

This construction requires an additional technical result; since it depends on some non-standard concentration bounds, we defer it to Appendix~\ref{app:prop33}.
\begin{proposition}
\label{conc-prop33}
If $\tau(H) \geq (r \log n)^{10}$, then there is an $O( \log r \log \Delta)$-round randomized algorithm to find a matching $M$ with $|M| \geq \Omega(\tau(H)/r)$ w.h.p.
\end{proposition}
\begin{proof}[Proof of Theorem~\ref{split4r}]
  The algorithm builds the matching $M$ over three phases. The first phase is the MIS algorithm of \cite{ghaffari-mis}, the second is a few iterations of Proposition~\ref{conc-prop33}, and the final phase is Theorem~\ref{mt2a}.

\textbf{Phase I.} We begin with the randomized part of the MIS algorithm of \cite{ghaffari-mis} applied to the line graph of $H$. This takes $O(\log(r \Delta))$ rounds, and w.h.p. it generates a matching $M$ such that every connected component in $\text{Res}_{M}(H)$ has size at most $\poly(r, \Delta) \log n$.

In the next two phases, all of the connected components will be handled independently. So consider an arbitrary component $H'$ of $\text{Res}_{M}(H)$; we need to find a maximal matching of $H'$. Initially, $\tau(H') \leq \poly(r, \Delta) \log n$ (since that is the maximum number of vertices in $H'$).

\textbf{Phase II.} We apply Proposition~\ref{conc-prop33} for $\Omega(r \log(r \Delta))$ stages; each time we do so, we get a matching of $H'$, which we commit to the matching $M$.  As long as $\tau( \text{Res}_M(H')) \geq (r \log n)^{10}$, each application of Proposition~\ref{conc-prop33} generates w.h.p. a matching of size $\Omega( \tau(\text{Res}_M(H'))/r)$. Thus, $\tau(\text{Res}_M(H'))$ shrinks by a $(1 - \Omega(1/r))$ factor and after $O( r \log (r \Delta))$ stages this reduces $\tau(\text{Res}_M(H'))$ to $(r \log n)^{10}$ w.h.p. Overall Phase II takes $O(r \log(r \Delta) \times \log r \log \Delta) = \tilde O(r \log^2 \Delta)$ rounds.

\textbf{Phase III.} We finish by applying \Cref{mt2a} to $\text{Res}_M(H')$ with $\delta = 1/\poly(n)$. Since the matching $M$ after Phase II  satisfies $\tau(\text{Res}_M(H')) \leq \poly(r, \log n)$, this runs in $\tilde O(r \log \log n \log \Delta + r^2 (\log \log n)^2 + r (\log \log n)^3)$ rounds. We obtain a maximal matching $M'$ of $\text{Res}_M(H')$; combined with $M$, this gives a maximal matching of $H'$.
\end{proof}

At this point there are a number of HMM algorithms available, which we summarize in the following result. In particular, this shows Theorem~\ref{mt2}.
\begin{theorem}
\label{sum-hmm-thm}
Consider a hypergraph $H$ with $n$ vertices, $m$ edges, rank $r$, maximum degree $\Delta$, and maximum matching size $\tau$. There are distributed HMM algorithms with the following complexities:
\begin{enumerate}
\item[(a)] $O(\log(r \Delta)) +  \polyloglog(m)$ rounds and failure probability $1/\poly(m)$.
\item[(b)] $O(\log m)$ rounds and failure probability $1/\poly(m)$.
\item[(c)] $\tilde O( (r \log \tau) (r \log \Delta + \log^2 \Delta) + \log^* n)$ rounds and failure probability zero.
\item[(d)] $\tilde O( (r \log \tau) (\log \Delta + r \log \log \tfrac{1}{\delta} + (\log \log \tfrac{1}{\delta})^2))$ rounds and failure probability $\delta$.
\item[(e)] $\tilde O(r \log^2 \Delta + r^2 (\log \log n)^2 + r (\log \log n)^3)$ rounds and failure probability $1/\poly(n)$.
\item[(f)] $\tilde O( (r \log n) (r \log \Delta + \log^2 \Delta))$ rounds and failure probability zero.
\end{enumerate}
\end{theorem}
\begin{proof}
Part (a) is the MIS algorithm of \cite{ghaffari-mis} combined with the network decomposition of \cite{ghaffari-roz} applied to the line graph of $H$. Part (b) is the MIS algorithm of \cite{luby} applied to the line graph of $H$. Parts (c) and (d) are restatements of Theorem~\ref{mt2a}. Part (e) is a restatement of Theorem~\ref{split4r}. Part (f) follows immediately from part (c), noting that $\tau \leq n$.
\end{proof}

One main motivation for HMM is in the context of graph algorithms. The following Theorem~\ref{mt4} lists some examples for edge-coloring. 
\begin{reptheorem}{mt4}
Let $G$ be a graph with maximum degree $\Delta$.
\begin{enumerate}
\item There is a $\tilde O(\log n \log^2 \Delta)$-round deterministic algorithm for $(2\Delta-1)$-list-edge-coloring.
\item There is a $\tilde O( (\log  \log n)^3 )$-round randomized algorithm  for $(2\Delta-1)$-list-edge-coloring
\item There is a $\tilde O(\Delta^4 \log^6 n)$-round deterministic algorithm for $\tfrac{3}{2} \Delta$-edge-coloring. 
\end{enumerate}
\end{reptheorem}
\begin{proof}
For (1),  Fischer, Ghaffari, and Kuhn \cite{fgk} reduces $(2\Delta-1)$-list-edge-coloring of a graph $G = (V,E)$ to maximal matching on a hypergraph of rank $r = 3$, with $O(|V|+|E|)$ vertices and maximum degree $O(\Delta^2)$. With these parameters, \Cref{sum-hmm-thm}(f) takes $\tilde O(\log^2 \Delta \log n)$ time.
  
For (2), there are a number of cases depending on the size of $\Delta$.  See \cite{fgk} for further details; the critical case is when $\Delta \leq \polylog n$, in which case Theorem~\ref{sum-hmm-thm}(e) takes $\tilde O( (\log \log n)^3 )$ rounds.

For (3), Ghaffari et al. \cite{GKMU18} describe a deterministic algorithm which uses HMM as a black box. Its runtime is $O(\Delta \log^3 n)$ plus $\Delta^2 \log n$ times the complexity of solving HMM on hypergraphs with $n$ vertices, rank $\Delta \log n$, and maximum degree $\Delta^{ O(r)}$. 

To find the HMM here, consider applying Theorem~\ref{sum-hmm-thm}(d) with parameter $\delta = 2^{-n^c}$ for a large constant $c$. This runs in $\tilde O(\Delta^2 \log^4 n)$ rounds, so it gives a randomized algorithm to find the desired edge-coloring, with run-time $\tilde O(\Delta^4 \log^6 n)$ and failure probability $2^{-n^c}$ for any desired constant $c$.

We can derandomize this by noting that there are at most $2^{\poly(n)}$ possibilities for the graph $G$ (including the ID's of all vertices). Since the randomized algorithm has failure probability $2^{-n^c}$, for $c$ sufficiently large there must exist a random seed which succeeds on \emph{all} such graphs $G$. Fixing this seed (which can be determined as a function of $n$ and $\Delta$) gives a deterministic algorithm.
\end{proof}

Previous algorithms for the first two problems \cite {prev} require  $O( \log^2 n \log^4 \Delta)$ and $O( (\log \log n)^6 )$ rounds respectively. The  previous algorithm for the third problem \cite{GKMU18} requires $\Delta^9 \polylog(n)$ rounds (the exponent of $\log n$ is not specified, but it is much larger than 6.)

\textbf{Remark.} Theorem~\ref{mt4}(3) illustrates a counter-intuitive phenomenon: it can be more efficient to use the \emph{randomized} HMM algorithm of Theorem~\ref{sum-hmm-thm}(d) with a very low failure probability in order to get \emph{deterministic} graph algorithms. The reason for this is that the deterministic HMM algorithm needs to succeed on \emph{all} rank-$r$ hypergraphs, whereas the randomized algorithm only has to succeed on a small fraction ($2^{n^2}$ out of a possible $2^{n^r}$) which come from graphs. 

\section{Approximate Nash-William decomposition}
\label{edge-sec}
We consider a variant of the classical Nash-Williams decomposition \cite{nw}. For a graph $G$,  there is an orientation of the edges such that every vertex has out-degree at most its arboricity. The  \emph{approximate edge-orientation} problem is to produce an edge-orientation where every vertex has out-degree at most $D = \lceil (1+\epsilon) \lambda \rceil$. (As usual in distributed algorithms, the parameter $\lambda$ should be viewed as a globally-known upper on arboricity.) Note that $\lambda \leq \Delta$, and it is possible that $\lambda \ll \Delta$.

We allow $G$ to be a multi-graph; by quantizing its adjacency matrix to multiples of $\poly(n/\epsilon)$, we can assume that $\lambda, \Delta \leq \poly(n, 1/\epsilon)$.

In \cite{GS17}, Ghaffari \& Su showed how to obtain such an orientation via a series of augmenting paths; they obtain a randomized algorithm running in $O(\log^4 n/\epsilon^3)$ rounds for simple graphs. This can be viewed as a HMM problem (wherein each augmenting path corresponds to a hyperedge). The deterministic HMM algorithm of \cite{fgk} converts this into a deterministic algorithm, which was subsequently improved by \cite{prev} to $O(\log^{10} n \log^5 \Delta / \epsilon^9)$ rounds. 

We also note that, after the original version of this paper, an alternate deterministic algorithm was developed based on LP solving instead of path augmentation \cite{hh2}; it runs in $\tilde O( \log^2 n / \epsilon^2)$ rounds. In particular, this strictly dominates our algorithm based on HMM.

Let us first summarize the algorithm of \cite{GS17}. The basic outline is to maintain an edge-orientation of $G$, and then iteratively improve it over stages $i = 1, \dots, \ell = \Theta( \frac{\log n}{\epsilon} )$. We let $G_i$ denote the resulting directed graph after stage $i$. The initial orientation $G_0$ can be arbitrary.  

At each stage $i$, we form an auxiliary graph $G'_i$ from $G_i$ by adding a source node $s$ and a sink node $t$. For each vertex $v \in G_i$ of out-degree $d > D$, we add $d - D$ edges from $s$ to $v$. For each vertex $v \in G_i$ of out-degree $d < D$, we add $D - d$ edges from $v$ to $t$. We then select a maximal set $P_i$ of edge-disjoint length-$i$ directed paths in $G'_i$ going from $s$ to $t$. We then ``augment'' the edge-orientation by reversing the orientation of all the edges in the paths in $P_i$.

The following result of \cite{GS17}, which is based on a blocking-path argument, explains why this overall process works.
\begin{theorem}[\cite{GS17}]
  \label{gs-thm}
The graph $G'_i$ has no $s-t$ paths of length strictly less than $i$. The graph $G_{\ell}$, for $\ell = O( \frac{\log n}{\epsilon} )$, has all its vertices with out-degree at most $D$.
\end{theorem}

In order to find $P_i$, we form an associated hypergraph $H_i$, whose edge set consists of all length-$i$ paths in $G'_i$ going from $s$ to $t$, and whose vertex set corresponds to all edges of $G'_i$. A maximal matching of $H_i$ is a maximal set of length-$i$ edge-disjoint paths in $G'_i$.

\begin{proposition}
Hypergraph $H_i$ has $\poly(n \lambda)$ vertices and at most $n^3 (2 \lambda)^{i}$ edges.
\end{proposition}
\begin{proof}
$G$ has $m \leq n \lambda$ edges. $H_i$ has a vertex for each of these, plus for each of the special edges leaving $s$ and coming to $t$. Each vertex of degree $d$ has at most $d$ special edges, so this contributes another factor of $m$ as well.

For the edge bound, let $U$ denote the set of vertices $v \in G_i$ with out-degree larger than $D$. We claim that any $s-t$ path $p$ in $G'_i$ contains at most one vertex $v$ in $U$. For, suppose that it contains two such vertices $v_1, v_2$, where $v_1$ comes before $v_2$. We could short-circuit this path, getting a path directly from $s$ to $v_2$ to $t$, which has length strictly less than $i$ in $G'_i$, contradicting Theorem~\ref{gs-thm}.

  Thus, in order to enumerate a directed path $p=(s, v_1, v_2, \dots, v_{i-1}, t)$ in $G'_i$, we begin by looping over the first edge from $s$ (which has at most $m$ choices), and the second vertex $v_2$ (which has at most $n$ choices.) For each $j = 3, \dots, i-1$, we have at most $D \leq 2 \lambda$ choices for $v_j$ since $v_j$ is an out-neighbor of $v_{j-1}$. Overall, we have $n^2 \lambda (2\lambda)^{i-3}$ choices for the path $p$.
    \end{proof}

\begin{reptheorem}{mt3}[The deterministic part]
  There is a deterministic $\tilde O(   \frac{ \log^6 n}{ \epsilon^4})$-round algorithm for approximate edge-orientation.
\end{reptheorem}
\begin{proof}
First consider the algorithm of Theorem~\ref{sum-hmm-thm}(d) to find the maximal matching of each hypergraph $H_i$, with failure probability $\delta = 2^{-(n/\epsilon)^c}$ for some constant $c$.  Note that $H_i$ has at most $b = n^3 (2 \lambda)^i$ edges, so it has maximum degree $\Delta \leq b$. Therefore, the HMM algorithm takes $\tilde O( \log \tau(H_i) (
i \log b + i^2 \log \log \tfrac{1}{\delta} + i (\log \log \tfrac{1}{\delta})^2) ) = \tilde O( \frac{ \log^4 n}{ \epsilon^2})$ rounds.

Now note that as $\lambda \leq \poly(n,1/\epsilon)$, there are at most $2^{\poly(n,1/\epsilon)}$ possibilities for the graph $G$. So for $c$ sufficiently large, we can choose a random seed which succeeds on all such graphs $G$. 

Since it requires $i \leq O( \frac{\log n}{ \epsilon} )$ rounds on $G$ to simulate a round on $H$, we find the HMM of hypergraph $H_i$ in $\tilde O(  \frac{\log^5 n}{ \epsilon^3})$ rounds. There are $\ell = O( \frac{\log n}{ \epsilon })$ rounds in total.
\end{proof}

We can get further advantage for the randomized algorithm by using sparsification.
\begin{reptheorem}{mt3}[The randomized part]
There is a randomized $\tilde O( \frac{\log^3 n}{\epsilon^3})$-round algorithm to compute an approximate edge-orientation w.h.p.
\end{reptheorem}
\begin{proof}
  We first get a randomized algorithm running in $\tilde O( \frac{\log^3 n \log \lambda}{ \epsilon^3})$ rounds. To do so, we use Theorem~\ref{sum-hmm-thm}(b) to get the maximal matching of each hypergraph $H_i$. Since $H_i$  has at most $n^3 (2 \lambda)^{\ell}$ edges,  this takes $O( \ell \log \lambda )$ rounds w.h.p.  Simulating $H_i$ takes $O(\ell)$ rounds with respect to $G$ and there are $O(\ell)$ stages.

  We next remove the $\log \lambda$ factor. If $\lambda \leq O( \frac{\log n}{\epsilon^2})$, then the $\log \lambda$ term is already hidden in the $\tilde O$ notation. Otherwise, randomly partition the edges as $E = E_1 \cup \dots \cup E_k$, for $k = \lceil \lambda/y \rceil$ classes, where $y = \frac{c \log n}{\epsilon^2}$ for a sufficiently large constant $c$.
  
  We claim that w.h.p, each graph $(V,E_j)$ has arboricity at most $\lambda' = y (1+\epsilon) + 1$. For, consider some edge-orientation $A$ of $G$ with out-degree at most $\lambda$. In the edge-orientation $A \cap E_j$, each vertex $v$ has at most $y$ outgoing edges in expectation.  Due to the size of $y$, the number of outgoing edges does not exceed $y (1+\epsilon)$ for any vertex w.h.p.
  
  We now run the previous randomized algorithm in parallel on each $(V,E_j)$ with parameter $\epsilon/10$, getting an edge-orientation of maximum out-degree $\lceil \lambda' (1+\epsilon/10)^2 \rceil$. If we combine all these edge-orientations, then any vertex has out-degree at most $k \lceil \lambda' (1+\epsilon/10)^2 \rceil$. For $\epsilon$ sufficiently small and $c$ sufficiently large, this is at most $\lambda ( 1 + \epsilon)$.
\end{proof}

\section{Acknowledgments}
Thanks to Mohsen Ghaffari and Fabian Kuhn for helpful discussions and reviewing some early drafts. Thanks to Hsin-Hao Su for discussions about Nash-Williams decompositions. Thanks to anonymous conference reviewers and journal reviewers for helpful comments.

\appendix

\section{Proof of Lemma~\ref{lem-split-crude}}
\label{degree-split-app}
We begin by using Lemma~\ref{basicderand2} for degree-splitting.
\begin{lemma}\label{lemma:HG_splitting}
Let parameters $\epsilon, \eta \in (0, \tfrac{1}{2})$ satisfy $\Delta \geq 100 \log(r/\eta)/\epsilon^2$. Given a good coloring of $H$, there is a deterministic $\tilde O(r \log(1/\eta)/\epsilon^2)$-round algorithm to generate  disjoint edge subsets $L_1, L_2 \subseteq E$ such that $a(L_1 \cup L_2) \geq (1-\eta) a(E)$ and $\deg_{L_j}(v) \leq (1+\epsilon) \Delta/2$ for all $v, j$.
\end{lemma}
\begin{proof}
Any vertex $v$ with $\deg(v) \leq \Delta/2$ will automatically have the degree condition satisfied, and can be ignored. Hence, we assume without loss of generality that $\deg(v) \geq \Delta/2$ for all vertices $v$.

  Define $\alpha = 50 \log( r/ \eta)/\epsilon^2$.  We construct a new hypergraph $H' = (E,U)$ by dividing every vertex $v \in V$ into $\ell = \left\lfloor\deg_H(v)/\alpha\right\rfloor $
  virtual nodes $u_1,\dots,u_{\ell}$ in $H'$ and we assign each of the
  hyperedges of $v$ to exactly one of the virtual nodes
  $u_1,\dots,u_{\ell}$, so that each $u_i$ has $\deg(u_i) \in [\alpha, 2 \alpha)$ and $\sum_{i=1}^{\ell} \deg(u_i) = \deg(v)$. This subdivision process is possible due to our assumption that every vertex in $V$ (that we are not ignoring) has degree at least $\alpha$. 
        
Our construction will have three parts. First, we define a function $F: U \rightarrow [0,\infty)$, which can be computed via a 1-round randomized algorithm on $H'$. Next, we derandomize this to select random bits such that $\sum_{u \in U} F_{u} \leq \sum_{u \in U} \bE[F_{u}]$. Finally, we construct $L_1, L_2$.
   
  To begin, randomly partition the edges into two sets $L_1', L'_2$, wherein each edge $e$ goes into $L_1'$ or $L'_2$ independently with probability $1/2$. For each virtual node $u \in U$ and $j = 1,2$, we define $Z_{u,j} = \deg_{L'_j}(u)$, and we set    
 $$
  F_{u} = a(N(u))  \bigl[ \bigl[  \bigvee_{j=1,2}Z_{u,j} > (1+\epsilon) \deg(u) /2 \ \bigr] \bigr]
  $$
  
Now let us compute $\sum_{u} \bE[F_{u}]$. For a node $u \in U$, the value $Z_{u,j}$ is a binomial random variable with mean $\deg(u)/2 \geq \alpha/2$. By the Chernoff bound, 
      $$
      \Pr( Z_{u,j} >(1+\epsilon) \deg(u) /2  ) \leq e^{-\epsilon^2 \deg(u)/6} \leq e^{-\epsilon^2 \alpha/16}
      $$

      Thus, overall we have
      \begin{align*}
        \sum_{u \in U} \bE[F_{u}] &\leq  \sum_{u \in U}\sum_j a(N(u)) \Pr( Z_{u,j} > (1+\epsilon) \deg(u) /2) \leq \sum_{u \in U} 2 a(N(u)) e^{-\epsilon^2 \alpha/16}
\end{align*}

      Since $H'$ has rank $r$, we have $\sum_{u \in U} a(N(u)) \leq r a(E)$, and so overall      
      $$
      \sum_{u} \bE[F_{u}] \leq a(E)( 2 r e^{-\epsilon^2 \alpha/16} ) \leq \eta a(E)
      $$
      where the last inequality follows from our choice of $\alpha$.
            
      By \Cref{basicderand2}, there is a deterministic $O(r \alpha)$-round algorithm to find random bits such that $\sum_{u} F_{u} \leq \sum_{u} \bE[F_{u}] \leq \eta a(E)$. Now, suppose we have fixed such random bits, determining the sets $L'_1, L'_2$. We form $L_1, L_2$ by starting with the sets $L_1', L_2'$ and then discarding any edges incident to a vertex $u \in U$ with $Z_{u,j} > (1+\epsilon) \deg(u)/2$.
      
     By summing the virtual nodes $u$ corresponding to a vertex $v \in V$, we get
      $$
      \deg_{L_j}(v) \leq \sum_{\substack{u \in U \\ \text{corresponding to $v$}}} (1+\epsilon) \deg(u)/2 \leq (1+\epsilon) \deg(v)/2 \leq (1+\epsilon) \Delta/2
      $$
      
      Furthermore, this ensures that
      $$
      a(L_1 \cup L_2) \geq a(E) - \sum_{u \in U} a(N(u)) [[ \bigvee_{j=1,2} Z_{u,j} > (1+\epsilon) \deg(u)/2]] = a(E) - \sum_{u \in U} F_{u}
      $$
	which is by construction at least $(1-\eta) a(E)$.
\end{proof}

We prove Lemma~\ref{lem-split-crude} by iterating this degree-splitting process:
\begin{proof}[Proof of Lemma~\ref{lem-split-crude}]
We will partition the edge sets over $s = \lfloor \log_2 k \rfloor$ stages, wherein each stage $i = 0, \dots, s-1$ has disjoint edge sets $T_{i,1}, \dots, T_{i,2^i}$. Specifically, we form the sets $T_{i+1,2j}, T_{i+1, 2j+1}$  at stage $i$ by applying Lemma~\ref{lemma:HG_splitting} in parallel to each hypergraph $H_{i,j} = (V, T_{i,j})$ with parameters $\epsilon = \frac{1}{4 s}, \eta = \frac{\delta}{4 s}$. Here $T_{i+1, 2j}$ and $T_{i+1, 2j+1}$ correspond to the edge-sets $L_1, L_2$.

  We first claim that $\Delta_i = (1 + \epsilon)^i \Delta/2^i$ is an upper bound on the degree of each $H_{i,j}$. We show this by induction. The base case $i=0$ is trivial. For the induction step, let us first show that the condition of Lemma~\ref{lemma:HG_splitting} is satisfied, namely
  $$
  \Delta_i \geq 100 \log(r/\eta)/\epsilon^2
  $$
  
  Since $\Delta_i =  (1 + \epsilon)^i \Delta/2^i \geq \Delta/2^s$ , it suffices to show that
  \begin{equation}
  \label{bbff3}
  \frac{\Delta}{2^s} \geq 100 \log(4 r s/\delta)/\epsilon^2
  \end{equation}

As $s = \lfloor \log_2 k \rfloor$, our hypothesis ensures that Eq.~(\ref{bbff3}) holds for $C$ sufficiently large. Therefore, Lemma~\ref{lemma:HG_splitting} ensures that every vertex $v$ has $\deg_{T_{i+1,j}}(v)  \leq (1+\epsilon) \Delta_i/2 = \Delta_{i+1}$, thus completing the induction. Furthermore,  we get $a(T_{i+1,2 j} \cup T_{i+1, 2j+1}) \geq a(T_{i,j}) (1-\eta)$ for every $i,j$.

Let $E_i = T_{i,1} \cup \dots \cup T_{i, 2^i}$. Summing over $j$ gives $a(E_{i+1}) \geq a(E_i) (1-\eta)$, which implies that
$$
a(E_s) \geq a(E) (1-\eta)^s \geq a(E) (1 - \delta)
$$

Now set $E' = E_s$ and set $\chi(e) = j$ for each $e \in T_{s,j}$. This ensures that every vertex $v \in V$ has $|N(v) \cap T_j| \leq \Delta_{s} \leq (1+\epsilon)^{s} \Delta/2^s \leq 2 \Delta/2^s  \leq 4 \Delta/k$ as desired.

In each stage,  Lemma~\ref{lemma:HG_splitting} runs in  $\tilde O(r \log(1/\eta)/\epsilon^2 ) = \tilde O(r s^2 \log \tfrac{1}{\delta})$ rounds. So the overall complexity is $\tilde O(r s^3 \log \tfrac{1}{\delta}) = \tilde O(r \log^3 k \log \tfrac{1}{\delta})$.
\end{proof}

\section{Proof of \Cref{find-fmatch1}}
\label{kmw-appendix}
We begin by considering an edge-weighting $a$ which has \emph{bounded} range; by quantizing edge weights later we remove this range dependence.

\begin{lemma}
  \label{maxmatch1a}
Let $H = (V,E)$ be a hypergraph with edge-weighting $a$ such that $W_{\min} \leq a(e) \leq W_{\max}$ for all edges $e$, and let $W = W_{\max}/W_{\min}$. There is a deterministic $O(\epsilon^{-4} \log(W r) \log(W \Delta))$-round  algorithm to find a fractional matching $h$ with $a(h) \geq (1-\epsilon) a^*(H)$.
\end{lemma}
\begin{proof}
The problem of finding a maximum-weight fractional matching $h$ can be interpreted as a type of packing LP, namely
  \begin{equation}
  \label{ggh0}
  \text{maximize $\sum_e a(e) h(e)$} \qquad \text{subject to $\forall v \sum_{e \in N(v)} h(e) \leq 1$}
  \end{equation}
  
  Kuhn, Moscibroda \& Wattenhofer \cite{kmw} provides a deterministic algorithm for solving such problems. Their analysis applies to generic packing LP's; however, it requires these to be parameterized in the following form:
  $$
  \text{maximize $\sum_i x(i)$} \qquad \text{subject to $x \in \mathbb R^n, A x \leq c$}
  $$
  All the entries of the constraint matrix $A$ must either have $A_{ij} = 0$ or $A_{ij} \geq 1$. Furthermore, if $A_{ij} > 0$ and $A_{i'j} > 0$ for two rows $i, i'$, the communications graph must have an edge from $i$ to $i'$.  (This parametrization is chosen so that the corresponding dual LP has a   nice structure.)
  
With this parametrization, the run-time of the algorithm of \cite{kmw} is determined by two key parameters $\Gamma_p$ and $\Gamma_d$ (here, $p$ and $d$ stand for primal and dual). These are defined as:
 $$
  \Gamma_p =  \bigl( \max_{j'} c_{j'} \bigr) \bigl( \max_j  \frac{ \sum_{i=1}^n A_{ij} }{c_j} \bigr), \qquad  \Gamma_d = \max_i \sum_{j=1}^m A_{ij}
  $$
With this parametrization, \cite{kmw} runs in time $O(\epsilon^{-4} \log \Gamma_p \log \Gamma_d)$ to get a $(1+\epsilon)$-approximation. 
  
  To transform the fractional matching LP into this form, we  define variables $x(e) = a(e) h(e)$ for each edge $e$, and our constraints become
  \begin{equation}
  \label{ggh1}
  \max \sum_{e} x(e) \qquad \text{subject to } \forall v \sum_{e \ni v} (W_{\max}/a(e)) x(e) \leq W_{\max}
  \end{equation}
  
  Given a solution $x$ to Eq.~(\ref{ggh1}), we will then set $h(e) = x(e)/a(e)$; this will clearly satisfy the fractional matching LP given by Eq.~(\ref{ggh0}). In this formulation, our constraint matrix $A$ is given by 
  $$
A_{ev} = \begin{cases}
  W_{\max}/a(e) & \text{if $e \in N(v)$} \\
  0 & \text{if $e \notin N(v)$} \\
  \end{cases}
  $$
  which has its entries either zero or in the range $[1,W]$.
 
 This LP now has the form required by \cite{kmw}, with the constraint vector $c$ having all its entries equal to $W_{\max}$. Therefore, we have 
	$$
        \Gamma_p = W_{\max} \max_{v \in V} \frac{ \sum_{e \in N(v)} (W_{\max}/a(e))}{W_{\max}} \leq W \Delta, \qquad \Gamma_d = \max_{e \in E} \sum_{v \in e} (W_{\max}/a(e)) \leq W r
$$
  and so the algorithm of \cite{kmw} runs in $O(\epsilon^{-4} \log(W \Delta) \log(W r))$ rounds.
\end{proof}

The next result describe deterministic and randomized methods to \emph{partially} discretize a fractional matching, while losing only a constant factor in the weight.

\begin{proposition}
  \label{maxmatch1}
  Let $H$ be a hypergraph with an edge-weighting $a$ and fractional matching $h'$ satisfying $a(h') \geq a^*(H)/2$.
  \begin{enumerate}
  \item There is a deterministic constant-round algorithm to generate an $10 \Delta$-proper fractional matching $h$ with $a(h) \geq \Omega(a^*(H))$.
  \item There is a randomized constant-round algorithm to generate a $\lceil 20 \log r \rceil$-proper fractional matching $h$ with $\bE[a(h)] \geq \Omega(a^*(H))$.
    \end{enumerate}
\end{proposition}
\begin{proof}
  
\noindent  (1) Form $h$ by rounding down $h'$  to the nearest multiple of $\delta = \frac{1}{10 \Delta}$. The loss incurred is at most $\delta a(E)$, i.e. $a(h) \geq a(h') - \delta a(E)$. By hypothesis, $a(h') \geq a^*(H)/2$. Also, we must have $a^*(H) \geq a(E)/\Delta$, as there is a trivial fractional matching (setting $a(e) = 1/\Delta$ for every edge) with this value. Thus,  $a(h) \geq a^*(H)/2 - 0.1 a(E)/\Delta \geq \Omega(a^*(H))$. 
  
\noindent (2) Consider the following random process:  generate edge multi-set $L'$ consisting of $X_e \sim \text{Poisson}(p(e))$ copies of each edge $e$, where  $p(e) = 10 h'(e) \log r$. If any vertex $v$ has more than $t = 20 \log r$ neighbors in $L'$, it discards all such neighbors. The remaining edges are the set $L$. 

It is clear that $\bE[a(L')] = 10 a(h') \log r$. Let us compute the probability that an edge $e$ gets removed from $L'$, due to some vertex $v \in e$ having a degree exceeding $t$. For every vertex $v \in e$, the value $Z = \deg_{L' - e}(v)$ is a Poisson random variable with mean at most $10 \log r = t/2$ (since $h'$ is a fractional matching). So Chernoff's bound gives $\Pr(Z >  t ) \leq e^{-10 \log r/3} \leq r^{-3.3}.$ A union bound over the vertices in $e$ shows that $e$ is removed from $L'$ with probability at most $r^{-2.3}$. 

Since this holds for every edge $e \in L'$, the expected weight of removed edges is at most $\sum_{e \in E} p(e) \times a(e) \times r^{2.3} \leq 3.7 a(h') \log r$. Summing over $e$ we thus get $\bE[a(L)] \geq 10 a(h') \log r - 3.7 a(h') \log r \geq \Omega(a^*(H) \log r)$. 

Thus, $h(e) = \frac{ [[e \in L]]}{\lceil 20 \log r \rceil}$ is a fractional matching $h$ which is $\lceil 20 \log r \rceil$-proper and which has $\bE[a(h)] \geq \Omega(\bE[\frac{a(L)}{\log r}]) \geq \Omega(a^*(H))$.
\end{proof}

Finally, we are ready to prove \Cref{find-fmatch1}.
\begin{proof}[Proof of \Cref{find-fmatch1}]
  Let us first show the deterministic algorithm. We let $E_i$ denote the set of edges with weights in the range $[(rq)^i, (rq)^{i+1})$ for parameter $q = 10 \Delta$, and let  $H_i = (V,E_i)$. Since we can decompose any fractional matching of $h$ into fractional matchings of the hypergraph $H_i$, we have $\sum_i a^*(H_i) \geq a^*(H)$.
    
    Our first step is to apply \Cref{maxmatch1a} to each $H_i$ to get a fractional matching $h'_i$ with $a(h'_i) \geq a^*(H_i)/2$. Since the edge weights are in the range $[(r q)^i, (r q)^{i+1}]$, this takes $O(\log^2 (\Delta r))$ rounds. We next use \Cref{maxmatch1} to get $q$-proper fractional matchings $h_i$ with $a(h_i) \geq \Omega(a^*(H_i))$.

We finally need to combine the fractional matchings $h_i$ into a single fractional matching $h$ for $H$. We form $h$ as follows: for any edge $e \in E_i$, if there is some $f \in E_j$ for $j \geq i+3$ with $h_j(f) > 0$ and $f \cap e \neq \emptyset$, then set $h(e) = 0$; otherwise,  set $h(e) = h_i(e)/3$. More formally:
  $$
  h(e) = \frac{h_i(e)}{3} [[ \bigwedge_{j = i+3}^{\infty} \bigwedge_{\substack{f \in E_j \\ f \cap e \neq \emptyset }} h_j(f) = 0]]
  $$

Clearly $h$ is $O(\Delta)$-proper. We need to check that $a(h) \geq \Omega( \sum_i a(h_i) )$ and that $h$ is  a fractional matching.  First, note that since $h_j$ is $q$-proper, we have $h_j(f) \geq 1/q$ whenever $h_j(f) > 0$. So:
  $$
  h(e) \geq \frac{h_i(e)}{3} \bigl( 1 - \sum_{j=i+3}^{\infty} \sum_{\substack{f \in E_j \\ f \cap e \neq \emptyset }} [[h_j(f) > 0]] \bigl) \geq \frac{h_i(e)}{3} \bigl( 1 - \sum_{j=i+3}^{\infty} \sum_{\substack{f \in E_j \\ f \cap e \neq \emptyset }} q h_j(f) \bigr)
$$
  We may now estimate $a(h)$ as:
  \begin{align*}
    a(h) &= \sum_i \sum_{e \in E_i} a(e) h(e) \geq \sum_i \sum_{e \in E_i} a(e) \frac{h_i(e)}{3} \bigl( 1 - \sum_{j=i+3}^{\infty} \sum_{\substack{f \in E_j \\ f \cap e \neq \emptyset }} q h_j(f) \bigr) \\
    &= \tfrac{1}{3} \sum_i \sum_{e \in E_i} h_i(e) a(e) - \tfrac{1}{3} \sum_j \sum_{f \in E_j} \frac{h_j(f)}{q} \sum_{i = -\infty}^{j-3} \sum_{\substack{e \in E_i \\ f \cap e \neq \emptyset}} a(e) h_i(e)
    \end{align*}
    
  For a given $j$ and $f \in E_j$, the second summand here can be rewritten as:
    \begin{align*}
 \frac{h_j(f)}{q} \sum_{i = -\infty}^{j-3} \sum_{\substack{e \in E_i \\ f \cap e \neq \emptyset}} a(e) &\leq     q h_j(f) \sum_{i = -\infty}^{j-3} \sum_{\substack{e \in E_i \\ f \cap e \neq \emptyset}} (r q)^{i+1} \times h_i(e) \qquad \text{since $a(e) \leq (r q)^{i+1}$} \\
    &\leq q h_j(f) \sum_{i = -\infty}^{j-3} (r q)^{i+1} \times |f| \qquad \text{since $h_i$ is a fractional matching} \\
    &\leq (r q) h_j(f) \times 2 (r q)^{j-2} \qquad \text{since $r q \geq 2$ and $|f| \leq r$}
    \end{align*}
  
Collecting terms, we get:    
    \begin{align*}
    a(h) &\geq \tfrac{1}{3} \sum_i \sum_{e \in E_i} h_i(e) a(e) - \tfrac{1}{3} \sum_j \sum_{f \in E_j} h_j(f) \times 2 (r q)^{j-1} = \tfrac{1}{3} \sum_i \sum_{e \in E_i} h_i(e) \bigl( a(e) - 2 (r q)^{i-1}  \bigr)
  \end{align*}
  
  For such an edge $e$, we have $a(e) \geq (r q)^i$, so the term $2 (r q)^{i-1}$ is at most $a(e)/2$. Thus, we have overall shown that
  $$
  a(h) \geq \tfrac{1}{3} \sum_i \sum_{e \in E_i} h_i(e) a(e)/2 \geq \Omega(\sum_i a(h_i))
  $$

To show $h$ is a fractional matching, consider some vertex $v \in V$, and let $j$ be maximal such there is an edge $f \in N(v) \cap E_j$ with $h_j(f) > 0$. We must then have $h(e) = 0$ for all $e \in N(v) \cap E_i$ for $i \leq j-3$. So we have:
  \[
  \sum_{e \in N(v)} h(e) \leq \sum_{i=j-2}^j \sum_{e \in N(v) \cap E_i} h(e) \leq \sum_{i=j-2}^j \sum_{e \in N(v) \cap E_i} h(e)/3 \leq 1.
  \]

The randomized algorithm is the same, except we set the parameter $q = \lceil 20 \log r \rceil$, and we use the randomized part of \Cref{maxmatch1} to get $\bE[a(h_i)] \geq \Omega(a^*(H_i))$. Thus, $\bE[a(h)] \geq \sum_i \bE[a(h'_i)] \geq \Omega(a^*(H))$. With this value of $q$, the complexity of \Cref{maxmatch1a} is only $O(\log r \log(\Delta r))$ rounds.
\end{proof}
  
  \section{Proof of Proposition~\ref{conc-prop33}}
  \label{app:prop33}
  To prove Proposition~\ref{conc-prop33}, we will use a general concentration inequality for polynomials developed in \cite{schudy-sviridenko}. We state this result in a (slightly specialized) form as follows:
  \begin{theorem}[\cite{schudy-sviridenko}]
    \label{ssthm}
    Consider a degree-$q$ polynomial function $S: \{0, 1 \}^n \rightarrow \mathbb R$ of the form
    $$
    S(y) = \sum_{ U \subseteq [n]} w_U \prod_{j \in U} y_j
    $$
    for non-negative weights $w_U$, where $w_U = 0$ for all sets $U$ with $|U| > q$.
    
    Suppose that $Y_1, \dots, Y_n$ are independent random variables, wherein each $Y_j$ is Bernoulli$(p_j)$, and let $Y = (Y_1, \dots, Y_n)$. For each $i = 0, \dots, q$, define parameter $\mu_q$ by
    $$
    \mu_i = \max_{ U \in \binom{[n]}{i}}  \sum_{ W: U \subseteq W \subseteq [n] } w_{W} \prod_{j \in W - U} p_j
    $$
    
    Then, there are constants $R_0, R_1 \geq 1$ such that
    $$
    \Pr \Bigl( |S(Y) - \mu_0| \geq \lambda \Bigr) \leq e^2 \max \Bigl( \max_{i=1, \dots, q} e^{-\frac{\lambda^2}{\mu_0 \mu_i R_0^i R_1^q}},  \max_{i=1, \dots, q} e^{-\bigl( \frac{\lambda}{\mu_i R_0^r R_1^q} \bigr)^{1/i}} \Bigr)
    $$
  \end{theorem}
     
  Using this bound, we now show Proposition~\ref{conc-prop33}. The algorithm we use is simple: we first apply \Cref{maxmatch1a} with the constant edge-weighting to get a fractional matching $h$ with $h(E)  \geq \Omega( \tau(H) )$.  We next form an edge-set $L$ by selecting each edge independently with probability $p_e = h(e)/(10 r)$. Finally, we form a matching $M$ from $L$ by discarding any intersecting edges.
  
  Let $b = h(E)$ and define the indicator variable $X_e = [[e \in L]]$ along with the related quantity
$$
S =\sum_{\substack{e, e' \in E \\ e \neq e' \\ e' \cap e \neq \emptyset}} X_e X_{e'}
$$

It is clear that $|M| \geq \sum_{e \in E} X_e - S$. As $\bE[\sum_{e \in E} X_e] = \frac{b}{10 r}$, a standard Chernoff bound calculation shows that $\sum_{e \in E} X_e \geq \frac{b}{20 r}$ with probability at least $1-e^{-\Omega(b/r)}$; by our assumption on $\tau(H)$, this is at least $1 - 1/\poly(n)$. We will show $S \leq \frac{b}{50 r}$ w.h.p.; this in turn will show that $|M| \geq \frac{b}{20 r} - \frac{b}{50 r} \geq \Omega( \tau(H)/r)$ w.h.p. as desired.

Observe that $S$ is a quadratic polynomial applied to the independent variables $X_e$, each of which is Bernoulli-$p_e$. This has precisely the form required for Theorem~\ref{ssthm}, wherein the weights are defined by setting $w_{U} = 1$ if $U$ consists of two overlapping edges  and $w_U = 0$ otherwise. The polynomial has degree $q = 2$ and so we need to calculate $\mu_0, \mu_1, \mu_2$.

Since all the weights are zero or one, we clearly have $\mu_2 \leq 1$. For $\mu_0$, we calculate:
\begin{align*}
\mu_0 &= \sum_{e \in E} \sum_{\substack{e' \neq e \\ e \cap e' \neq \emptyset}} p_e p_{e'} \leq \sum_{e \in E} p_e \sum_{v \in e} \sum_{e' \in N(v)} \frac{h(e')}{10 r} \leq \sum_{e \in E} \frac{p_e}{10 r} = \frac{b}{100 r}
\end{align*}

Similarly, we may calculate $\mu_1$ as:
\begin{align*}
\mu_1 &= \max_e \sum_{e': e \cap e' \neq \emptyset} p_{e'} \leq \max_e \sum_{v \in e} \sum_{e' \in N(v)} \frac{h(e')}{10 r} \leq \frac{1}{10}
\end{align*}

We apply Theorem~\ref{ssthm} with parameter $\lambda = \frac{b}{100 r}$ to get:
\begin{equation}
\label{gtag1}
\Pr( S \geq \frac{b}{50 r} ) \leq \Pr( |S - \mu_0| \geq \lambda) \leq e^2 \max \Bigl( e^{-\frac{\lambda^2}{\mu_0 \mu_1 R_0 R_1^2}}, e^{-\frac{\lambda^2}{\mu_0 \mu_2 R_0^2 R_1^2}}, e^{-\frac{\lambda}{\mu_1 R_0 R_1^2}}, e^{-\frac{\lambda^{1/2}}{\mu_2^{1/2} R_0 R_1}} \Bigr)
\end{equation}

Since $R_0, R_1$ are constants and $\mu_1, \mu_2 \leq O(1)$ and $\mu_0 \leq \lambda$, the four terms in the RHS of Eq.~(\ref{gtag1}) are bounded by respectively $e^{-\Omega(\lambda)}, e^{-\Omega(\lambda)}, e^{-\Omega(\lambda)}, e^{-\Omega(\lambda^{1/2})}$. Since $\lambda = \frac{b}{50 r} \geq \frac{(r \log n)^{10}}{50 r} \geq \Omega(\log^3 n)$, each of these is in turn at most $n^{-\omega(1)}$. So $S \leq \frac{b}{50 r}$ w.h.p. as we have claimed.

\end{document}